\documentclass[journal]{IEEEtran}

\ifCLASSINFOpdf

\else

\fi

\usepackage{amsthm}
\usepackage{amssymb}
\usepackage{amsmath}
\usepackage{graphicx}
\usepackage{subfigure}
\usepackage{caption}
\usepackage{pmat}
\usepackage{float}
\usepackage{lipsum}
\usepackage{enumitem}
\usepackage{epstopdf}
\usepackage{multirow}
\usepackage{tabularx}
\usepackage{lipsum}
\usepackage{multicol}
\usepackage{amsfonts}
\usepackage[english]{babel}
\usepackage[autostyle]{csquotes}
\usepackage{array}
\usepackage{dcolumn}
\theoremstyle{remark}
\usepackage{booktabs}

\usepackage{arydshln}
\usepackage{siunitx}

\usepackage{algorithm}
\usepackage[end]{algpseudocode}
\usepackage{setspace}
\usepackage{nicefrac}     
\usepackage{microtype}
\usepackage[bookmarks=false]{hyperref}
\usepackage{hyperref}

\usepackage{mathtools, cuted}
\usepackage{flushend}

\begin{document}
 \thispagestyle{empty} 
\onecolumn

\vspace*{10cm}
\textbf{This work  submitted to journal/IEEE transaction for possible publication. Copyright may be transferred without notice, after which this version may no longer be accessible}

\newpage
\title{Gap Reduced Minimum Error Robust Simultaneous Estimation For Unstable Nano Air Vehicle}

\author{Jinraj~V.~Pushpangathan,~
        Harikumar~Kandath,~\IEEEmembership{Member,~IEEE,},
          Suresh~Sundaram,~\IEEEmembership{Senior Member,~IEEE,},
        and~Narasimhan~Sundararajan,~\IEEEmembership{Life Fellow,~IEEE}
\thanks{Research fellow @ Department
of Aerospace Engineering, Indian Institute of Science, Bangalore-560012,
India, e-mail: (jinrajaero@gmail.com).}
\thanks{Assistant professor @ International Institute of Information Technology, Hyderabad, India, e-mail: ( harikumar.k@iiit.ac.in)}
\thanks{Associate professor @ Department
of Aerospace Engineering, Indian Institute of Science, Bangalore- 560012, India, e-mail: (vssuresh@iisc.ac.in).}
\thanks{Professor (Retd.) @ School of Electrical and Electronics Engineering, Nanyang Technological University, Singapore, e-mail: (ensundara@ntu.edu.sg).}
}

\markboth{}%
{Shell \MakeLowercase{\textit{et al.}}: Bare Demo of IEEEtran.cls for IEEE Journals}

\twocolumn
\maketitle

\begin{abstract}

This paper proposes a novel Gap Reduced Minimum Error Robust Simultaneous (GRMERS) estimator for resource-constrained Nano Aerial Vehicle (NAV) that enables a single estimator to provide simultaneous and robust estimation for a given $N$ unstable and uncertain NAV plant models. The estimated full state feedback enables a stable flight for NAV. The GRMERS estimator is implemented utilizing a Minimum Error Robust Simultaneous (MERS) estimator and  Gap Reducing (GR) compensators. The MERS estimator provides robust simultaneous estimation with minimal largest worst-case estimation error even in the presence of a bounded energy exogenous disturbance signal. The GR compensators reduce the \textit{gap} between the \textit{graphs} of $N$ linear plant models to decrease the estimation error generated by the MERS estimator. A sufficient condition for the existence of a simultaneous estimator is established using LMIs and robust estimation theory. Further, MERS estimator and GR compensator design are formulated as non-convex tractable optimization problems and are solved using the population-based genetic algorithms. The performance of the GRMERS estimator consisting of MERS estimator and GR compensators from the population-based genetic algorithms is validated through simulation studies. The study results indicate that a single GRMERS estimator can produce state estimates with reduced errors for all flight conditions. The results indicate that the single GRMERS estimator is robust than the individually designed  $H_{\infty}$ filters.
\end{abstract}


\IEEEpeerreviewmaketitle

\section{Introduction}
\raggedbottom

\noindent Recent trends in  Micro Air Vehicles (MAVs) point to the development of a new class of  small air vehicles called Nano Air Vehicles (NAVs) that execute specific missions undetected with a high degree of agility.  NAVs can be broadly classified into three categories, viz., fixed-wing NAVs, rotary wing NAVS, and flapping-wing NAVs. They are widely used for   intelligence operations, battlefield surveillance,  reconnaissance, and disaster assessment missions. These small vehicles have severe dimensional and weight constraints as their overall dimensions and weights need to be lower than $75$~mm and $20$~g, respectively \cite{jinsmc}. Figure \ref{fig:foo} shows  a typical fixed-wing NAV that weighs $19.4$~g and has an overall dimension of $75$~mm \cite{jin}.

\noindent In general, the plant models of these NAVs are multi-input-multi-output (MIMO), unstable, uncertain, 
adversely coupled, and have a different number of unstable poles \cite{jinthesis,nogar}. In this paper, for convenience,  we use the terms \textit{plant} and \textit{estimator} to represent the \textit{plant model} and \textit{estimator} model, respectively. More details about these plant characteristics are given in \cite{jinsmc,jin}. Due to these  plant characteristics, the NAV's require  complex feedback controllers to  accomplish a mission. For using the existing closed-form solutions of the full state feedback for designing a feasible controller for a NAV requires the measurements or estimates of all the system’s state variables. However, due to the weight and dimensional constraints, the autopilot hardware of a NAV like the one shown in Fig. \ref{fig:autopilot} \cite{jin} has severe resource constraints, such unavailability of lightweight sensors to measure every state variable (like translational velocities, angle-of-attack, airspeed, side slip angle). Moreover, these autopilots have both limited computational and memory powers. Hence, one needs to design a computationally simple full state estimator from the available measurements.
\begin{figure}[h!]
	\centering
	\subfigure[Fixed-wing NAV \label{fig:foo}]{\includegraphics[width=1.4in, height=0.9in]{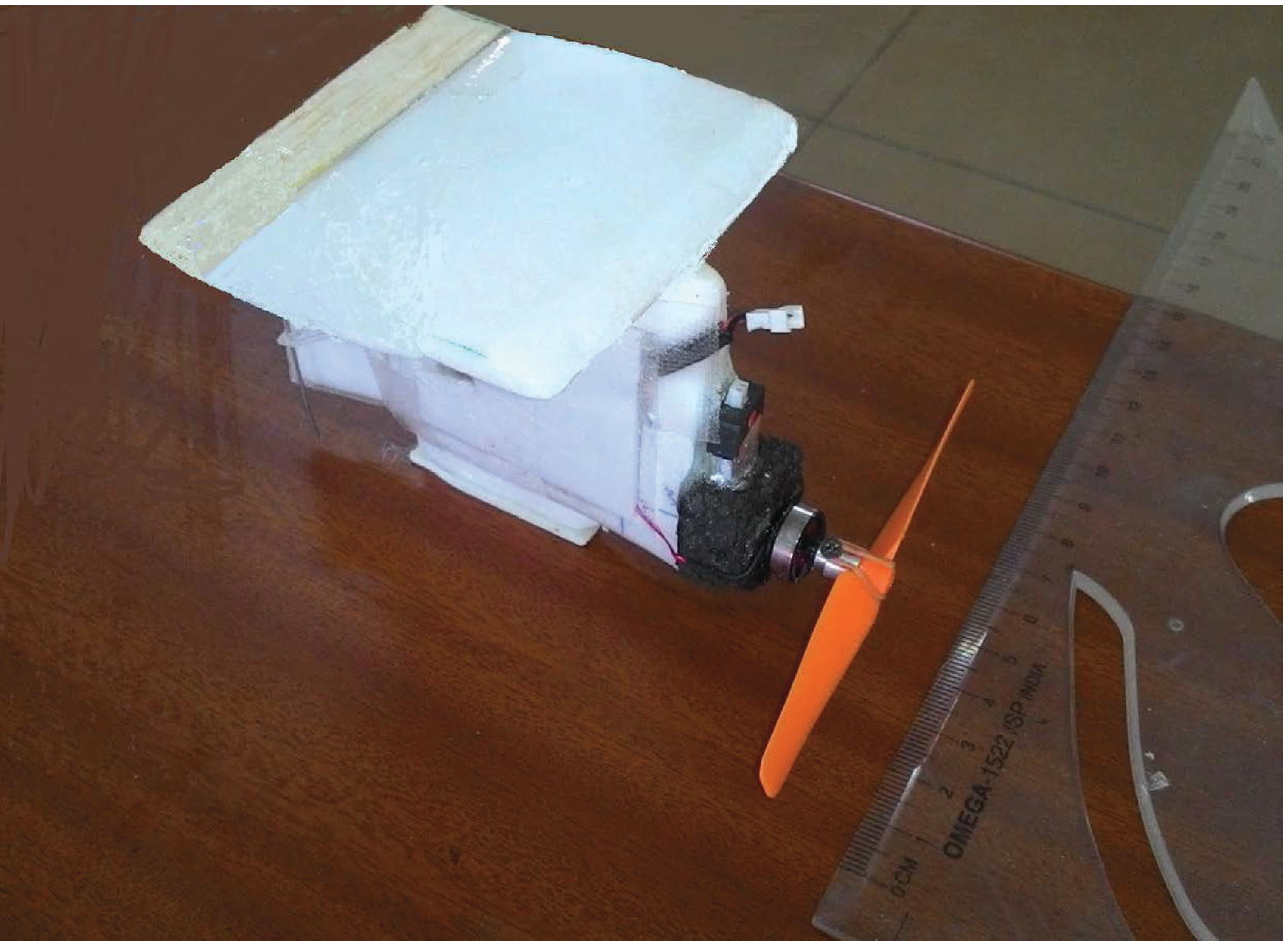}}
	\subfigure[Autopilot hardware  \label{fig:autopilot}]{\includegraphics[width=1.4in, height=0.9in]{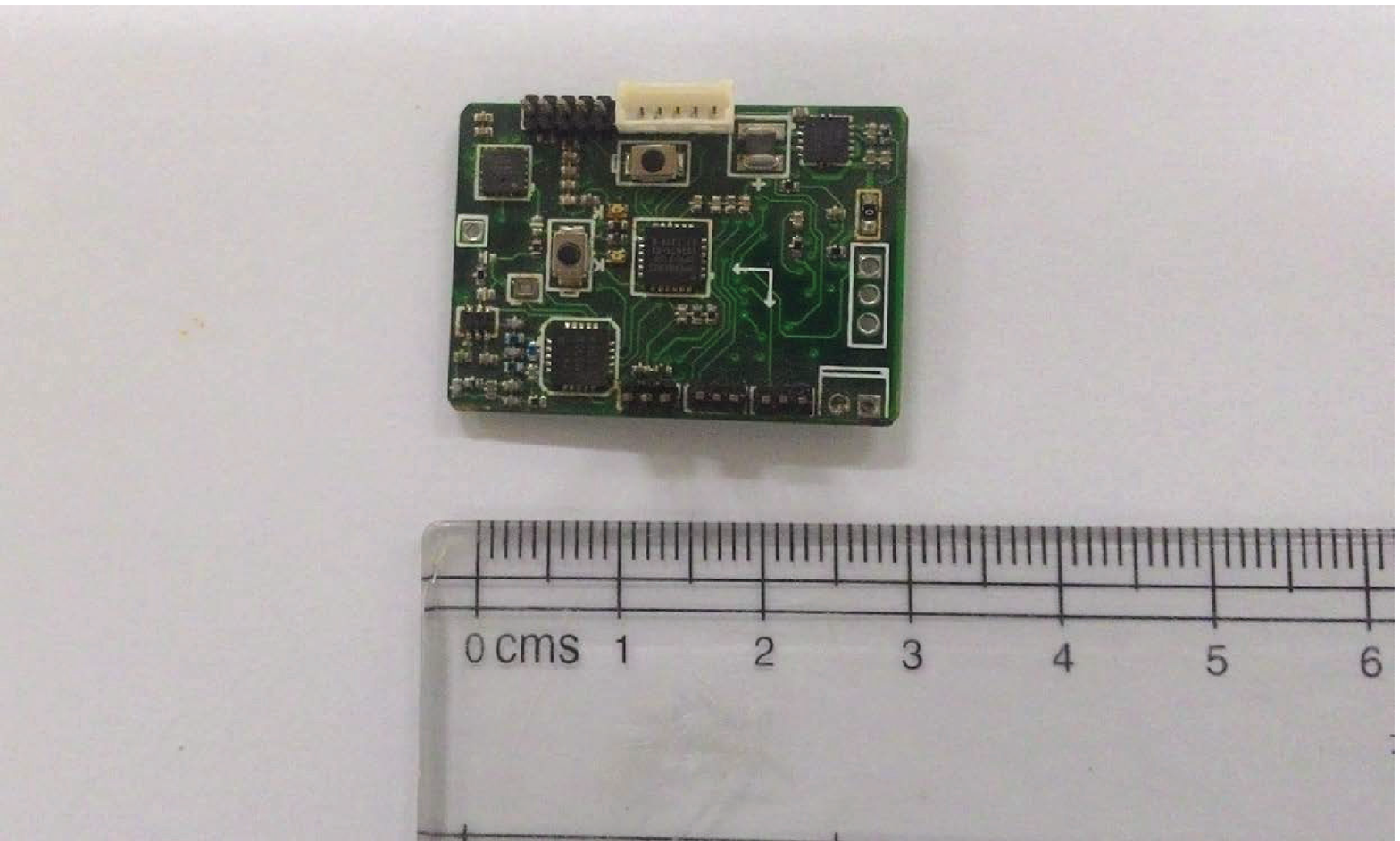}}
	\caption{75~mm wingspan NAV and  autopilot hardware}
\end{figure}
Extended Kalman Filter (EKF) is the de-facto standard for UAV estimation schemes. However, EKF is not suitable for a NAV, as its autopilot hardware has limited computational and memory resources, specifically for Jacobian computing. A  gain-scheduled EKF reduces the computational cost of calculating the estimator gain \cite{pham,Hork}, but it still requires the measurement of scheduling variables like airspeed etc.  Further, the  significant model uncertainties in NAVs can induce notable errors in the state  estimates. These difficulties clearly point out to the need for an estimation algorithm that is computationally simpler and robust to model uncertainties. This algorithm should also cater to both stable/unstable plant models and  should not require a computationally expensive gain-scheduling approach. To put the problem in a sharp focus, a NAV requires a single computationally less intensive and robust estimator to estimate the states of a finite set of MIMO LTI uncertain unstable plants. Herein,  such an estimator is  referred to  as a Robust Simultaneous (RS) estimator. The overall problem definition of simultaneous estimation is given below.

\textit{Simultaneous Estimation problem}: Let us consider a finite set,  $\mathcal{P}=\{\mathbf{P}_i(s) \in \mathcal{RL_\infty}~|~ i\in\{1,\dots, l, \dots, N\}$, containing stabilizable and detectable MIMO LTI plant models in a transfer function matrix form (i.e.,    $\mathbf{P}_i(s) \in \mathcal{RL_\infty}$). Here,  $\mathcal{RL_\infty}$ symbolizes the space of proper, real-rational functions of $s\in \mathcal{C}$ which has no  poles in the imaginary axis of $s$-plane. The state-space form of $\mathbf{P}_i(s)$ is given by
\begin{equation}
\begin{aligned}
\mathbf{P}_i(s): \left \{ 
\begin{array}{ll}
\dot{\mathbf{x}}_i(t) & = A_{i}\mathbf{x}_i(t)+B_i\mathbf{u}_{i}(t)  \\
\mathbf{y}_i(t) & = C_i\mathbf{x}_i(t)+\mathbf{v}(t)     \\
\mathbf{z}_i(t) & =C_z\mathbf{x}_i(t)
\end{array}
\right .
\label{eqpa1}
\end{aligned}
\end{equation}
where $\mathbf{x}_i(t) \in \mathbb{R}^{\hat{n}}$,  $\mathbf{u}_i(t) \in \mathbb{R}^{\hat{m}}$, $\mathbf{y}_i(t) \in \mathbb{R}^{\hat{r}}$, and $\mathbf{z}_i(t) \in \mathbb{R}^{\hat{q} \leq\hat{n}}$
represent the state vector, the control input vector, the measurement  vector, and the   vector   that contains those states that need to be estimated,  respectively.
Besides this, $\mathbf{v}(t) \in \mathbb{R}^{\hat{r}}$ is a bounded energy  measurement noise vector. Note that the noise  considered in this article is zero-mean Gaussian white. Further, $A_i \in \mathbb{R}^{\hat{n} \times \hat{n}}$, $B_i \in \mathbb{R}^{\hat{n} \times \hat{m}} $, and $C_i \in \mathbb{R}^{\hat{r} \times \hat{n}}$    are the system, input, and output matrices, respectively. Also,  $C_z \in \mathbb{R}^{\hat{q} \times \hat{n}}$ is a constant matrix for all the plants. Let $\hat{\mathbf{P}}_l(s) \in \mathcal{RH}_\infty$ be a suitable estimator for all the plant models belonging to $\mathcal{P}$. Here, $\mathcal{RH_\infty}$ denotes the space of proper, real-rational functions of $s\in \mathcal{C}$ which are analytic in $\mathcal{C}_+$. Let the state-space form  of $\hat{\mathbf{P}}_l(s)$ be formed using the state-space matrices of $\mathbf{P}_l(s) \in \mathcal{P}$. Following this, one of the state-space realizations of $\hat{\mathbf{P}}_l(s)$ is  given by 
\begin{eqnarray}
\hat{\mathbf{P}}_l(s): \left \{
\begin{array}{ll}
\dot{\hat{\mathbf{x}}}_l(t) = & A_l\hat{\mathbf{x}}_l(t) + B_l\mathbf{u}(t) + \mathbf{L}_l \left(\mathbf{y}(t)-C_l\hat{\mathbf{x}}_l(t)\right)\\
\hat{\mathbf{z}}_{l}(t)= & C_z\hat{\mathbf{x}}_{l}(t)
\end{array}
\right.
\label{eqest}
\end{eqnarray}
where $\hat{\mathbf{x}}_{l}(t) \in \mathbb{R}^{\hat{n}}$ is the state vector of the estimator,  $\hat{\mathbf{z}}_{l}(t) \in \mathbb{R}^{\hat{q}}$ is the output vector of the estimator which provides the estimate of $\hat{\mathbf{z}}_{i}(t)$ of $\mathbf{P}_i(s) \in \mathcal{P}$ for all $i~ \in \{1,\dots,l,\dots,N\}$, and   $\textbf{L}_l \in \mathbb{R}^{\hat{n} \times \hat{r}}$ is the estimator gain. In (\ref{eqest}), $\mathbf{u}(t)$ and $\mathbf{y}(t)$ are the inputs to the estimator  from the plant. For example, $\mathbf{u}(t)=\mathbf{u}_l(t)$ and $\mathbf{y}(t)=\mathbf{y}_l(t)$ when $\hat{\mathbf{P}}_l(s)$ becomes the estimator of $\mathbf{P}_l(s)$. Also, for all $i \in \{1,\dots,l,\dots, N\}$,  $\mathbf{e}_{\mathbf{z}_{il}}(t)=\mathbf{z}_{i}(t)-\mathbf{\hat{z}}_{l}(t) \in \mathbb{R}^{\hat{q}}$ be the estimation error vector when  $\hat{\mathbf{P}}_l(s)$  estimates $\mathbf{z}_{i}(t)$ of  $\mathbf{P}_i(s) \in \mathcal{P}$, respectively. Let $\mathbf{d}(t)$ be a bounded energy exogenous signal vector that adversely affect the estimation error dynamics of the estimator by  increasing the estimation errors (e.g. $\mathbf{v}(t)$ given in (\ref{eqpa1})). 
Now, the simultaneous estimation problem is about finding an estimator $\hat{\mathbf{P}}_l(s)$  for  $\mathcal{P}$ such that the following two conditions hold $ \forall ~i \in \{1,\dots, l, \dots,N\}$.\\
\textit{Condition I}: when $\mathbf{d}(t)=\mathbf{0}$, all the estimation error vectors should  asymptotically converge to zero, i.e.,
\begin{align}
\lim_{t \to\infty} \mathbf{e}_{\mathbf{z}_{il}}(t)&=0
 \label{cd-1}
\end{align}
\textit{Condition II}: when $\mathbf{d}(t)\neq\mathbf{0}$, the root mean square (RMS) gain from $\mathbf{d}(t)$ should be bounded, i.e.,
\begin{align}
\underset{\underset{\mathbf{d}(t)\neq\mathbf{0}}{\mathbf{d}(t) \in \mathcal{L}_2[0,\infty)}}{\sup}~\frac{\parallel \mathbf{e}_{\mathbf{z}_{il}}(t) \parallel_2}{\parallel \textbf{d}(t) \parallel_2}& < \gamma \label{cd-2} 
\end{align}
where $\gamma$ is a constant that satisfies $1>\gamma >0$ and $\sup$ symbolizes supremum. In this paper,   only \textit{Condition II} is considered as the measurements of NAV are always affected by  $\mathbf{v}(t)$.
\noindent Before formulating the problem of simultaneous estimation for a NAV, we provide first a brief review of the existing  simultaneous estimation research work in the literature below.\par 

\noindent The problem of simultaneous observation was first studied in \cite{yao}. In this paper, coprime factorization technique was utilized to solve the simultaneous observation problem. For the finite set that contains at least one stable plant, the  necessary and sufficient conditions for the existence of simultaneous observations were obtained.   Using the proposed method, a simultaneous observer for two plants  was designed.  In \cite{kov}, a stable inverse approach was employed to synthesize a simultaneous observer for a given set of plants. The restrictions on these  plants were that they should not have any right half plane zeros besides  satisfying the condition,  $(\hat{m}+\hat{r})>\hat{m}N$. For plants that have the same number of inputs and outputs without common eigenvalues, the necessary and sufficient conditions for the existence of a simultaneous functional observer were presented in \cite{jai}. In \cite{lau},  algebraic geometry tools were presented to characterize the simultaneous observability of a set of linear single-input single-output plants and also to design a simultaneous state observer for the same. The  methods presented in \cite{yao}-\cite{lau} are not suitable for  synthesizing   a simultaneous estimator for a NAV due to the following reasons:  All  the plants of the NAV may be unstable  and also can have zeros on the right half of the $s$-plane \cite{jinjgcd}. Also, the number of outputs of the NAV is   more than the number of inputs violating the condition mentioned in \cite{jai}. Furthermore, in the case of a NAV,   satisfying the condition, $(\hat{m}+\hat{r})>\hat{m}N$ mentioned in \cite{kov} is not possible. For example,  a NAV with three inputs and five outputs, a simultaneous estimator can be synthesized only for two plants. The measured outputs of the plants of a  NAV are also affected by noise. Apart from this, the plants of the NAV are subjected to higher model uncertainties. In \cite{yao,lau}, no method is explicitly proposed to provide the desired performance (to achieve the \textit{condition II})   by the estimator when the plants are subjected to measurement noise and higher model uncertainties. To overcome the above-mentioned limitations, one needs to develop a new method to synthesis a simultaneous estimator for a NAV.\par
In this paper, we propose a novel   Gap Reduced Minimum Error  Robust Simultaneous (GRMERS) estimator to handle unstable plants with model uncertainties and measurement noises. The GRMERS estimator incorporates the solution of two problems: a Minimum Error Robust Simultaneous (MERS) estimation problem and a Gap Reduced (GR) compensator problem.  The  MERS estimation problem  finds a single estimator referred to as the MERS estimator that accomplishes robust state  estimation with minimal (largest) worst-case estimation error for $N$ number of unstable uncertain MIMO linear plants of the NAV. The estimation error of the MERS estimator is further reduced by decreasing the  \textit{gap}  (see \cite{vinfre} for definition) between the \textit{graphs}  (see Definition \ref{defgra}) of the plants in $\mathcal{P}$  by cascading these plants with suitable pre/post compensators.  These compensators are called the GR compensators, and the corresponding synthesis problem is termed as the GR compensators problem. The GR compensators are defined by first-order differential equations which  can be solved by using the limited computational capabilities of the NAV's autopilot hardware. Using the robust estimation theory, the  MERS estimator  and GR compensator designs problems are devised as non-convex optimization problems following the robust estimation theory, formulated in terms of Linear Matrix Inequalities (LMI) and the properties of $v$-gap metric, respectively.
The major highlights  of the proposed GRMERS estimator  in this paper are: 
\begin{enumerate}
\item 	This approach can handle the robust simultaneous (RS) estimation of more than three ($N>$3) minimum/non-minimum phase unstable plants even having common eigenvalues.
\item	To our best knowledge, it has been  shown here (for the first time) that  cascading the plants with compensators (GR compensators) that reduce the gap between the graphs of the plants can reduce the root mean square value of estimation errors of a simultaneous estimator.
\end{enumerate} 
The effectiveness of both the MERSE and GRC algorithms are demonstrated by generation of the MERS estimator  and the GR compensators to synthesise the GRMERS estimator for four unstable plants of the NAV mentioned in \cite{jin}. The stability,  nominal, and robust performances of the GRMERS estimator are 1) validated through numerical simulations with Gaussian measurement noise and 2) compared with the performances of MERS estimator and $H_\infty$ filter. For this purpose, individual   $H_\infty$ filters are designed separately for each plant. The nominal performance analysis indicates that the best performance is given by individual $H_\infty$ filters, followed by the GRMERS estimator.  As compared to the MERS estimator, the GRMERS estimator yields up to $55\%$ reduction in estimation error, which substantiates the effectiveness of providing GR compensators. The robust performance analysis, however, shows that the GRMERS estimator has a lower estimation error of up to $43~\%$ compared to the $H_\infty$ filters.

\noindent The paper is organized as follows.   The dynamics of a fixed-wing NAV and the problem formulations are presented  in Section \ref{P_S}. In Section \ref{SRS}, the   GRMERS estimator is discussed. The design and performance evaluation of  the  GRMERS estimator  are presented in Section \ref{DPE}. Finally,  Section \ref{cons} summarizes the key results of this paper.

\section{Problem Formulations of Minimum Error Robust Simultaneous Estimator and Gap Reduced Compensators for a Fixed-wing NAV}\label{P_S}

\noindent In this section,   the dynamic model of a fixed-wing NAV along with its complexities is described  first  to motivate the need for  robust simultaneous estimation for a NAV . Next,  the precise mathematical problem formulations for the design of the Minimal Error Robust Simultaneous (MERS) estimator and the Gap Reducing (GR) compensators are presented.

\subsection{Dynamics of a Single Propeller Fixed-Wing NAV}\label{D_N}

\noindent Here, a  brief description of the dynamic model for  a single propeller fixed-wing NAV is provided, and more details can be found in \cite{jin,jinsmc}.  Generally, in a fixed-wing aircraft,  the actuator$'$s bandwidth would be much higher than that of the plant, whereas   it is not  true in the case of a NAV.  Hence, in the flight controller and estimator design of a NAV, one has to explicitly include the dynamics of the actuator along with the plant dynamics.  Besides,  the dynamics of a single propeller fixed-wing NAV  has significant cross-coupling effects. Based on these considerations, a suitable linear model for a single propeller fixed-wing NAV with both the coupling effects and actuator dynamics is the linear coupled model \cite{jin} given by
\begin{equation}
  \dot{\mathbf{x}}_i(t)={A}_{i}\mathbf{x}_i(t)+B\mathbf{u}_{i}(t) \label{adcp3}  
\end{equation}
where $\mathbf{x}_i(t) \in \mathbb{R}^{\hat{n}=11}$, $\mathbf{u}_{i}(t) \in \mathbb{R}^{\hat{m}=3}$, $A_{i} \in \mathbb{R}^{(11 \times 11)}=\left[\begin{array}{c;{2pt/2pt}r}A_{{Lo}_i}&A_{{Lo}_i}^{La}\\\hdashline[2pt/2pt]A_{{La}_i}^{Lo}&A_{{La}_i} \end{array}\right]$, and $B_{i} \in\mathbb{R}^{(11 \times 3)}=\left[\begin{array}{c;{2pt/2pt}r}B_{{Lo}_i}&B_{{Lo}_i}^{La}\\\hdashline[2pt/2pt]B_{{La}_i}^{Lo}&B_{{La}_i} \end{array}\right]$ are the state vector, the control input vector, the system matrix, and the control matrix, respectively. Here, $A_{{Lo}_i} \in \mathbb{R}^{6 \times 6}$, $A_{{La}_i} \in \mathbb{R}^{5 \times 5}$, $B_{{Lo}_i} \in \mathbb{R}^{6 \times 2}$, and $B_{{La}_i} \in \mathbb{R}^{5 \times 1}$ represent the system and control matrices of the longitudinal and lateral state-space models, respectively. Also, $A_{{Lo}_i}^{La}  \in \mathbb{R}^{6 \times 5}$, $A_{{La}_i}^{Lo}  \in \mathbb{R}^{5 \times 6}$,  $B_{{Lo}_i}^{La}  \in \mathbb{R}^{6 \times 1}$, and $B_{{La}_i}^{Lo}  \in \mathbb{R}^{5 \times 2}$ are the longitudinal coupling block of $A_i$, lateral coupling block of $A_i$, longitudinal coupling block of $B_i$, and lateral coupling block of $B_i$, respectively.  Furthermore,   $\mathbf{x}_{i}(t)$ and  $\mathbf{u}_{i}(t)$ in (\ref{adcp3}) are defined as
\begin{align}
\mathbf{x}_{i}(t) =&\left[\begin{array}{ccccccccccc}
u & w&q&\theta&\delta_{e}&\delta_{T}&v & p&r&\phi&\delta_{r}
\end{array}\right]^T(t)  \label{sts}\\
\mathbf{u}_{i}(t) =&\left[\begin{array}{ccc}
\delta_{eu}& \delta_{Tu}&\delta_{ru}
\end{array}\right]^T(t) \label{ipts}
\end{align}
where $\left[\begin{array}{ccc}
u(t) & v(t)&w(t)
\end{array}\right]^T$ is the body-fixed linearized translational velocities in m/s, $\left[\begin{array}{ccc}
p(t) & q(t)&r(t)
\end{array}\right]^T$ is the body-fixed linearized rotational velocities in rad/s, and $\left[\begin{array}{cc}
\theta(t) & \phi(t)
\end{array}\right]^T$ is the body-fixed linearized Eulers angles in rad. Also, $
\delta_{e}(t)$, $\delta_{r}(t)$, and $\delta_T(t)$ are the linearized elevator deflection (in rad), rudder deflection  (in rad), and propeller speed (in rps-revolution per second), respectively. In (\ref{ipts}),  $\delta_{eu}(t)$ (rad), $\delta_{ru}(t)$ (in rad), and $\delta_{Tu}(t)$ (in rps) represent the inputs to  the elevator actuator, the input to the rudder actuator, and the input to the electric motor  that drives the propeller, respectively.

\noindent The linear dynamics   of a fixed-wing NAV is adversely coupled, uncertain, and unstable as seen from dynamics of  the  $75$~mm wingspan fixed-wing NAV mentioned in  \cite{jin}. Hence, the NAVs similar to  the $75$~mm wingspan NAV require  flight controllers to handle all these complexities and  accomplish the desired  mission. This controller can use  either  full state feedback or output feedback strategy. Generally, the full state feedback strategy is preferred as various closed-form solutions are available when compared with the output feedback strategy. The development of  a well-proven  full state feedback flight controller   requires  the measurement of all the state variables.   If  all the state variables can not be measured,  then  estimates of unmeasured states are required. Among all the state variables of the NAV,  the  measured  state variables are $q(t)$, $\theta(t)$, $p(t)$, $r(t)$, and  $\phi(t)$ and can be directly used for control. Hence,  the  measurement vector of the NAV, $\mathbf{y}_i(t) \in \mathbb{R}^{\hat{r} =5}$,  is given by
\begin{align}
\mathbf{y}_{i}(t) =&\left[\begin{array}{cccccccc}
q&\theta& p&r&\phi
\end{array}\right]^T(t) 
\label{Meqn}
\end{align}
Following this, the measurement equation of the NAV is given by
\begin{align}
\mathbf{y}_i(t)=C\mathbf{x}_i(t)+\mathbf{v}(t)
\label{oteqn}
\end{align}
where $\mathbf{v}(t) \in \mathbb{R}^{5}$ and $C \in \mathbb{R}^{5 \times 11}$. Due to the absence of lightweight sensors for measurement, $u(t), v(t), $  $w(t)$, $\delta_{e}$, $\delta_{T}$, and $\delta_{r}$ need to be estimated.  Thus, the estimation vector for a  NAV, $\mathbf{z}_{i}(t)  \in  \mathbb{R}^{\hat{q}=6}$, is given by 
\begin{align}
\mathbf{z}_{i}(t) =&\left[\begin{array}{cccccc}
u & w&\delta_{e}&\delta_{T}&v&\delta_{r} 
\end{array}\right]^T(t)
\label{outeqn}
\end{align}
Then, the equation of $\mathbf{z}_i(t)$ is given by 
\begin{align}
\mathbf{z}_{i}(t) =&C_z\mathbf{x}_i(t)
\label{outeqn1}
\end{align}
where $C_z \in \mathbb{R}^{6 \times 11}$. Also, note that, in the case of a NAV, $\hat{n}$=$\hat{r}+\hat{q}$. Consequently, the state-space model  of the NAV used for  designing  the estimator  is given by (\ref{eqpa1}) with $B_i$=$B$,  $C_i$=$C$, $\hat{n}=11$, $\hat{m}=3$, $\hat{r}=5$,  and $\hat{q}=6$. The resource  constrained autopilot hardware  and the uncertain and unstable nature of NAV's LTI plants  suggest that there is a need for a MERS estimator for the estimation of $\mathbf{z}_{i}(t)$. The description of the MERS estimator is given in the next subsection.

\subsection{Minimum Error Robust Simultaneous Estimation Problem}

\noindent To describe the MERS estimation problem, consider $\mathcal{P} \subset	 \mathcal{RL_\infty}$ containing $N$ number of stabilizable and detectable LTI MIMO unstable adversely coupled  uncertain plants of the NAV given  in  \cite{jin}.  The  state-space form of any plant  belonging  to $\mathcal{P}$ is given by (\ref{eqpa1}) with $B_i$=$B$,  $C_i$=$C$, $\hat{n}=11$, $\hat{m}=3$, $\hat{r}=5$,  and $\hat{q}=6$.  Let  $\hat{\mathbf{P}}_l(s) \in \mathcal{RH}_{\infty}$  represents an estimator.  The estimator, $\hat{\mathbf{P}}_l(s)$ is formed using the state-space matrices of $\mathbf{P}_l(s) \in \mathcal{P}$ and a suitable estimator gain, $\textbf{L}_l \in \mathbb{R}^{\hat{n}=11 \times \hat{r}=5}$. Following this, 
the state-space model of $\hat{\mathbf{P}}_l(s)$ is  given by (\ref{eqest})
with $\mathbf{\hat{x}}_{l} \in \mathbb{R}^{\hat{n}=11}$  and $\mathbf{\hat{z}}_{l}(t) \in \mathbb{R}^{\hat{q}=6}$. Here, $\mathbf{\hat{z}}_{l}(t)$  provides the estimate of $\mathbf{z}_{i}(t)~ \forall~ i~ \in \{1,\dots,l,\dots,N\}$. 
When we consider $\mathbf{\hat{P}}_l(s)$ as the common estimator of  $\mathbf{P}_i(s) \in \mathcal{P}~ \forall~i \in \{1,\dots, l, \dots, N\}$, then the state estimation error vectors  and the estimation error vectors are denoted by $\mathbf{e}_{\mathbf{x}_{il}}(t)=\mathbf{x}_i(t)-\mathbf{\hat{x}}_{l}(t)$ and  $\mathbf{e}_{\mathbf{z}_{il}}(t)=\mathbf{z}_i(t)-\mathbf{\hat{z}}_{l}(t)$ for all $i~\in \{1,\dots, l, \dots, N\}$, respectively.  Similarly, for all $i~\in~\{1, \dots, l, \dots, N\}$,  $\mathbf{e}_{\mathbf{y}_{il}}(t)=\mathbf{y}_i(t)-C\mathbf{\hat{x}}_{l}(t)$.
\par 
Now, consider the case where $\mathbf{\hat{P}}_l(s)$ is employed  to estimate $\mathbf{z}_{i}(t)~ \forall ~ i \in \{1, \dots, l, \dots,N\}$ of $\mathbf{P}_i(s) \in \mathcal{P}~ \forall~i \in \{1, \dots, l, \dots,N\}$, respectively. Then, one can obtain $N$ number of estimator error models that are given by 
 \begin{equation}
 \begin{aligned}
 \dot{\mathbf{e}}_{\mathbf{x}_{il}}(t)={}&A_l\mathbf{e}_{\mathbf{x}_{il}}(t)+\Delta A_{il} \mathbf{x}_{i}(t)+B_e\mathbf{u}_e(t)\\
 \mathbf{e}_{\mathbf{z}_{il}}(t)={}&C_z\mathbf{e}_{\mathbf{x}_{il}}(t)\\
 \mathbf{e}_{\mathbf{y}_{il}}(t)={}&C\mathbf{e}_{\mathbf{x}_{il}}(t)+\mathbf{v}(t); i \in \{1, \dots, l, \dots, N\}
 \end{aligned}
 \label{spes_1} 
 \end{equation}
where $\Delta A_{il}$ is the difference between the system matrices, $A_l$ and $A_i$ of $\mathbf{P}_l(s)$ and $\mathbf{P}_i(s)$, respectively, $\mathbf{x}_{i}(t)$ is the state vector of $\mathbf{P}_i(s)$, $B_e=\mathbf{I} \in \mathbb{R}^{11 \times 11}$ is the input matrix of error dynamics, and $\mathbf{u}_e(t)=-\mathbf{L}_l \mathbf{e}_{\mathbf{y}_{il}}(t)$. Applying this  $\mathbf{u}_e(t)$ in (\ref{spes_1})  results in
\begin{align}
\dot{\mathbf{e}}_{\mathbf{x}_{il}}(t)=(A_l-B_eC\mathbf{L}_l)\mathbf{e}_{\mathbf{x}_{il}}(t)+\Delta A_{il} \mathbf{x}_{i}(t) - B_e\mathbf{L}_l\mathbf{v}(t)
\label{tf2}
 \end{align}
 Equation (\ref{tf2}) suggests that, along with $\mathbf{v}(t)$,  $\mathbf{x}_{i}(t)$ also becomes  an exogenous signal vector that adversely affect the estimation error dynamics. This is because of the difference between the system matrices of $\mathbf{P}_i(s)$ and $\mathbf{\hat{P}}_l(s)$. Hence, when $\mathbf{\hat{P}}_l(s)$ estimates $\mathbf{z}_i(t)$ of $\mathbf{P}_i(s) \in \mathcal{P}$, then $\mathbf{d}(t) \triangleq \left[\begin{array}{cc}
\mathbf{x}_i(t) &\mathbf{v}(t)
\end{array}\right]^T$.

\noindent Note that $\mathbf{x}_{i}(t)$ needs to be a bounded energy  exogenous signal. For that either $\mathbf{P}_i(s)$ or its closed-loop plant  must to be stable. Further, when $\mathbf{\hat{P}}_l(s)$ becomes the common estimator of all the plants belonging to $\mathcal{P}$, then there exist $N$ number of closed-loop transfer function matrices, $\textbf{e}(s)_{\mathbf{P}_i\mathbf{\hat{P}}_l} \in  \mathcal{RH_\infty} $, from   $\mathbf{d}(s) \triangleq \left[\begin{array}{cc}
\mathbf{x}_i(s) &\mathbf{v}(s)
\end{array}\right]^T$ to $\textbf{e}_{\textbf{z}_{il}}(s)$ for all $i ~\in~\{1,\dots, l, \dots, N\}$. Now, consider $N$ number of   estimators, $\mathbf{\hat{P}}_1(s), \dots, \mathbf{\hat{P}}_i(s),\dots,\mathbf{\hat{P}}_l(s),\dots,\mathbf{\hat{P}}_N(s)$, each formed using state-space matrices of  $\mathbf{P}_1(s) \in \mathcal{P},\dots,\mathbf{P}_i(s) \in \mathcal{P}, \dots,\mathbf{P}_l(s) \in \mathcal{P}, \dots, \mathbf{P}_N(s) \in \mathcal{P}$, respectively. Let  the finite set,  $\hat{\mathcal{P}}$, contains all these $N$ estimators. The state-space forms of these estimators are given as
 \begin{equation}
 \mathbf{\hat{P}}_1(s): \begin{cases}
 \dot{\mathbf{\hat{\mathbf{x}}}}_{l}(t)&=A_l\mathbf{\hat{x}}_l(t)+B\mathbf{u}(t)+\mathbf{L}_l (\mathbf{y}(t)-C\mathbf{\hat{x}}_l(t))\\
 \mathbf{\hat{z}}_{l}(t)&=C_z\mathbf{\hat{x}}_{l}(t); \forall~ l \in \{1,\dots, i, \dots, N\}
 \end{cases}
 \label{eqest2}
 \end{equation}
where  $\mathbf{u}(t)$ $\in   \{\mathbf{u}_1(t),\dots,\mathbf{u}_i(t),\dots,\mathbf{u}_N(t)\}$  and $\mathbf{y}(t) \in \{\mathbf{y}_1(t),\dots,\mathbf{y}_i(t),\dots,\mathbf{y}_N(t)\}$. Here, $\mathbf{u}_i(t)$ and $\mathbf{y}_i(t)$ are the control input and measurement   vectors of $\mathbf{P}_i(s) \in \mathcal{P}$, respectively.
Further, when each estimator  belonging to $\hat{\mathcal{P}}$ is utilized as the common estimator for all the $N$ number of plants of $\mathcal{P}$, then  $N \times N$ closed-loop transfer function matrices, $\textbf{e}(s)_{\mathbf{P}_i\mathbf{\hat{P}}_l}$  from $\textbf{d}(s)$ to $\textbf{e}_{\textbf{z}_{il}}(s)~ \forall~ i, l \in \{1,\dots, N\}$ are obtained.  Let the $H_\infty$ norm of $\textbf{e}(s)_{\mathbf{P}_i\mathbf{\hat{P}}_l}$ that  provides the worst-case gain from  $\mathbf{d}(t)$ is defined  as
   \begin{align}
  ||\textbf{e}(s)_{\mathbf{P}_i\mathbf{\hat{P}}_l}||_\infty=\underset{\underset{\mathbf{d}(t)\neq\mathbf{0}}{\mathbf{d}(t) \in \mathcal{L}_2[0,\infty)}}{\sup}~\frac{\parallel \mathbf{e}_{\mathbf{z}_{il}}(t) \parallel_2}{\parallel \textbf{d}(t) \parallel_2}
  \label{htf}
  \end{align}
Now, if we consider $\mathbf{\hat{P}}_l(s) \in \hat{\mathcal{P}}$ as the simultaneous estimator of all the plants belonging to $\mathcal{P}$, then    $\parallel \textbf{e}(s)_{\mathbf{P}_i\mathbf{\hat{P}}_l} \parallel_\infty ~ \forall~ i \in \{1,\dots, l ,\dots,N\}$ are the worst-case gains from $\mathbf{d}(s)$ associated with $\mathbf{\hat{P}}_l(s)$.  Also, the largest worst-case gain from $\mathbf{d}(s)$ associated with $\mathbf{\hat{P}}_l(s)$ is $\max~\{\parallel \textbf{e}(s)_{\mathbf{P}_i\mathbf{\hat{P}}_l} \parallel_\infty\}~ \forall~ i \in \{1,\dots, l ,\dots,N\}$. The largest worst-case gain from $\mathbf{d}(s)$ associated with an estimator  belongs to $\hat{\mathcal{P}}$ occurs while estimating the desired state variables of a plant referred to as the worst plant  from the perspective of the simultaneous estimation process. This worst plant   is indicated through the  subscript $k$  along with a superscript  that shows its association with the corresponding estimator. Following this, the largest worst-case gain from $\mathbf{d}(s)$ associated with  $\hat{\mathbf{P}}_l(s) \in \hat{\mathcal{P}}$, $\parallel \mathbf{e}(s)_{\mathbf{P}_{k}^{l}\mathbf{\hat{P}}_l} \parallel_\infty$, is defined  as
\begin{align}
\parallel \mathbf{e}(s)_{\mathbf{P}_{k}^{l}\mathbf{\hat{P}}_l} \parallel_\infty=\max \{||\mathbf{e}(s)_{\mathbf{P}_i\mathbf{\hat{P}}_l}||_\infty ~|~i\in\{1,\dots, l ,\dots,N\}\}
\label{cond2}
\end{align}  
 From (\ref{cd-2}) and (\ref{htf}), the sufficient conditions for considering  $\mathbf{\hat{P}}_l(s) \in \hat{\mathcal{P}}$ as a simultaneous estimator of all the plants belonging to $\mathcal{P}$ with respect to \textit{Condition II} are $\mathbf{e}(s)_{\mathbf{P}_i\mathbf{\hat{P}}_l} < \gamma ~ \forall~ i \in \{1,\dots, l, \dots, N\}$. An equivalent single condition that satisfy  the above conditions can be obtained  using (\ref{cond2}), and is given by
\begin{equation}
\parallel \mathbf{e}(s)_{\mathbf{P}_{k}^{l}\mathbf{\hat{P}}_l} \parallel_\infty < \gamma
\label{cond1}
\end{equation}
Now, the MERS estimation problem can be defined precisely as:  Given $\mathcal{P}$.
 Find   $\mathbf{\hat{P}}_j(s) \in \hat{\mathcal{P}}$ along with  $\mathbf{L}_j$ such that 1) $\mathbf{\hat{P}}_j(s)$ simultaneously estimates $\mathbf{z}_i(t)$ of $\mathbf{P}_i(s) \in \mathcal{P}~ \forall~ i~\in~\{1,\dots, j, \dots, N\}$ even when these plants have model  uncertainties and 2) the condition given by
 \begin{equation}
 ||\mathbf{e}(s)_{\mathbf{P}_{k}^j\mathbf{\hat{P}}_j}||_\infty =  \min  \left\{ ||\mathbf{e}(s)_{\mathbf{P}_{k}^l\mathbf{\hat{P}}_l}||_\infty~|~l \in \{1,..., N\}\right\}  < \gamma
 \label{cond3}
 \end{equation}
  is satisfied.
  
\noindent The solution to this problem is a RS estimator referred to as the MERS estimator whose  $\parallel \mathbf{e}(s)_{\mathbf{P}_{k}^{j}\mathbf{\hat{P}}_j} \parallel_\infty$  is the smallest among $\big\{ ||\mathbf{e}(s)_{\mathbf{P}_{k}^l\mathbf{\hat{P}}_l}||_\infty~|~l \in \{1,\dots, j, \dots, N\}\big\}$ with $\parallel \mathbf{e}(s)_{\mathbf{P}_{k}^{l}\mathbf{\hat{P}}_l} \parallel_\infty < \gamma$. This suggests that suggest that the
largest worst-case estimation error of MERS estimator is the
minimal among the largest worst-case estimation errors of 
simultaneous estimators.
In this paper, we  consider only parametric  uncertainties in the form of bounded perturbations in the system matrices.
The performance of the MERS estimator can be further improved by appending suitable compensators to the plant dynamics. This is discussed in the next section.

\subsection{Gap Reducing Compensator Problem}
\noindent The Gap Reducing (GR) compensator problem is about  finding those compensators that modify the input-output characteristics of  all the plants belonging to $\mathcal{P}$ for reducing  further the estimation errors arising from the differences between the system matrices of the MERS estimator  and $\mathbf{P}_i(s)$  $\forall~i \in \{1, \dots, j, \dots, N\}$ ($\Delta A_{ij}$). The GR compensator problem can be stated as: assume that there exists a MERS estimator, $\mathbf{\hat{P}}_j(s) \in \mathcal{\hat{P}}$, for the plants belonging to $\mathcal{P}$. Now, consider  the plant, $\mathbf{P}_j(s) \in \mathcal{P}$ and assume $\mathcal{P}$ as an uncertainty model set with $\mathbf{P}_j(s)$ as the nominal plant. Besides this, assume also that the plants belonging to $\mathcal{P} \setminus {\mathbf{P}_j(s)}$ as the perturbed plants of ${\mathbf{P}_j(s)}$. Let  $\mathbf{N}_j(s) \in \mathcal{RH_\infty}$ and $\mathbf{M}_j(s) \in \mathcal{RH_\infty}$ with det($\mathbf{M}_j(s) \neq 0$) are the normalized right coprime factors of $\mathbf{P}_j(s)$. Subsequently, $\mathbf{P}_j(s)$ is given by
\begin{align}
\mathbf{P}_j(s)=&\mathbf{N}_j(s)\mathbf{M^{-1}}_j(s)
\label{copr1}
\end{align}
Let $\mathbf{\Delta}_{N_{\mathbf{P}_j\mathbf{P}_i}}(s) \in \mathcal{RH_\infty}$ and  $\mathbf{\Delta}_{M_{\mathbf{P}_j\mathbf{P}_i}}(s) \in \mathcal{RH_\infty}$ are the  right coprime factor perturbations of $\mathbf{N}_j(s)$ and $\mathbf{M}_j(s)$ with $\big|\big|[\begin{array}{cc}\mathbf{\Delta}_{N_{\mathbf{P}_j\mathbf{P}_i}}(s)&  \mathbf{\Delta}_{M_{\mathbf{P}_j\mathbf{P}_i}}(s)\end{array}]^{T}\big|\big|_\infty\leq\epsilon_{{}_{\mathbf{P}_j\mathbf{P}_i}}$, respectively. Here,   $\epsilon_{{}_{\mathbf{P}_j\mathbf{P}_i}}$  is  the least upper bound on the  right coprime factor perturbations. Then, $\mathbf{P}_i(s)~ \forall~ i \in \{1,\dots,N\} \setminus {j}$ are defined as
\begin{align}
\mathbf{P}_i(s)=&\big(\mathbf{N}_j(s)+\mathbf{\Delta}_{N_{\mathbf{P}_j\mathbf{P}_i}}(s)\big)\big(\mathbf{M}_j(s)+ \mathbf{\Delta}_{M_{\mathbf{P}_j\mathbf{P}_i}}(s)\big)^{-1}\\
& \forall~ i \in \{1,\dots,N\} \setminus {j} \nonumber
\label{coprb}
\end{align}
Now, consider the scenario where   $\mathbf{v}(t)=0$ and $\mathbf{\hat{P}}_j(s)$ is employed to estimate $\mathbf{z}_l(t)$ of $\mathbf{P}_j(s)$. To realize this scenario, the state vector of the plant  needs to be bounded. For that, a feedback controller is employed as shown in Fig. \ref{sest1}. Now, assume $\mathbf{P}_j(s)$ is perturbed to form $\mathbf{P}_i(s) \in \mathcal{P} \setminus \{\mathbf{P}_j(s)\}$ when $\mathbf{\hat{P}}_j(s)$ estimates $\mathbf{z}_l(t)$ of $\mathbf{P}_j(s)$. In that case, $\mathbf{\hat{P}}_j(s)$ receives $\mathbf{u}_i(t)=   \mathbf{u}_j(t)+\Delta  \mathbf{u}_i(t)$ and $\mathbf{y}_i(t)=   \mathbf{y}_j(t)+\Delta  \mathbf{y}_i(t)$ instead of $\mathbf{u}_j(t)$ and $\mathbf{y}_j(t)$, respectively as shown in Fig. \ref{sest1}. Here, $\Delta  \mathbf{u}_i(t)$ and $\Delta  \mathbf{y}_i(t)$ are the perturbations in the inputs and outputs of $\mathbf{P}_j(s)$, respectively arising from $\Delta A_{ij}$. So when we develop the estimation error dynamics offline, the plant  considered is $\mathbf{P}_j(s)$. But in reality, the inputs to $\mathbf{\hat{P}}_j(s)$ will be $\mathbf{u}_i(s)$ and $\mathbf{y}_i(s)$ of the perturbed plant of $\mathbf{P}_j(s)$. Following this, the estimation error dynamics is given by 
\begin{equation}
\begin{aligned}
\dot{\mathbf{e}}_{x_{jj}}(t)=&(A_j - B_e \mathbf{L}_j C) \mathbf{e}_{x_{jj}}(t) -B \Delta  \mathbf{u}_i(t)  -  \mathbf{L}_j \Delta  \mathbf{y}_i(t) \\
=&(A_j - B_e \mathbf{L}_j C) \mathbf{e}_{x_{jj}}(t) -B (\mathbf{u}_i(t)-\mathbf{u}_j(t))\\ &  -  \mathbf{L}_j (\mathbf{y}_i(t)-\mathbf{y}_j(t)) \\
\mathbf{e}_{z_{jj}}(t)=&C_z\mathbf{e}_{x_{jj}}(t)
\end{aligned}
\label{ed4}
\end{equation}
Let the eigenvalues of $(A_j - B_e \mathbf{L}_j C)$ belong to $\mathbb{C}_-$. Then,  (\ref{ed4}) suggests that $\mathbf{e}_{z_{jj}}(t)$ converges to zero when $\mathbf{u}_i(t)$=$\mathbf{u}_j(t)$ and $\mathbf{y}_i(t)$=$\mathbf{y}_j(t)$. Following this, $\mathbf{e}_{z_{jj}}(t)$will be closer to zero if we make $\mathbf{u}_i(t)$ and $\mathbf{y}_i(t)$ closer to $\mathbf{u}_j(t)$ and  $\mathbf{y}_j(t)$, respectively. This indicates that when $(\mathbf{u}_i(t) , \mathbf{y}_i(t))~ \forall~ i \in \{1,\dots,N\} \setminus \{j\}$ become closer to $(\mathbf{u}_j(t), \mathbf{y}_j(t))$, then the estimation error due to $\Delta A_{ij} \mathbf{x}_i(t) ~\forall~i \in \{1,\dots,N\}\setminus \{j\}$ becomes close to zero. 

\noindent To state the GR compensator problem,  consider the following definition.
\newtheorem{defn}{Definition}[section]
\begin{defn}
	\textit{Graph, $\mathcal{G}$, of an operator:}~~{\normalfont Let $\mathbf{P}_j(s) : \mathcal{H} \rightarrow \mathcal{H} $ be any linear operator in Hilbert space ($\mathcal{H}$) defined on the domain, $\mathcal{D}(\mathbf{P}_j(s)) \subseteq \mathcal{H}$. Then, the \textit{graph} of $\mathbf{P}_j(s)$, $\mathcal{G}(\mathbf{P}_j)$ is defined as
		\begin{align}
		\begin{split}
		\mathcal{G}(\mathbf{P}_j) & =\{(\mathbf{u}_j(s) , \mathbf{y}_j(s)) \in \mathcal{H} \times \mathcal{H} : \mathbf{u}_j(s) \in \mathcal{D}(\mathbf{P}_j(s)), \mathbf{y}_j(s) = \\ & \qquad \mathbf{P}_j(s)\mathbf{u}_j(s) \in \mathcal{H}\}
		\end{split}
		\end{align}
		
	}
\label{defgra}
\end{defn}
\noindent This definition suggests that  $	\mathcal{G}(\mathbf{P}_j)$  is the set of all pairs of $(\mathbf{u}_j(s) , \mathbf{y}_j(s))$ with $\mathbf{u}_j(s) \in \mathcal{D}(\mathbf{P}_j(s))$. Also,  this definition indicates  that making $\mathcal{G}(\mathbf{P}_i)~ \forall ~i \in \{1,\dots,N\}\setminus \{j\}$ closer to $\mathcal{G}(\mathbf{P}_j)$ increases the closeness between $(\mathbf{u}_i(t) , \mathbf{y}_i(t))~ \forall~ i \in \{1,\dots,N\} \setminus \{j\}$ and $(\mathbf{u}_j(t), \mathbf{y}_j(t))$.  Cascading the plants with pre and post compensators, $\mathbf{W_{in}}(s) \in \mathcal{RH_\infty}$  and $\mathbf{W_{ot}}(s) \in \mathcal{RH_\infty}$, respectively, we modify the input-output characteristics and thereby the \textit{graphs} of the plants. Now, the GR compensator problem can be stated as:  Find $\mathbf{W_{in}}(s) \in \mathcal{RH_\infty}$  and $\mathbf{W_{ot}}(s) \in \mathcal{RH_\infty}$ such that
\begin{enumerate}
	\item 
 $\max\{gap(\mathcal{G}(\mathbf{\acute{P}}_j(s)), \mathcal{G}(\mathbf{\acute{P}}_i(s))~|~ i \in {1,\dots, j, \dots, N}\}$ is lower than  $\max\{gap(\mathcal{G}(\mathbf{P}_j(s)), \mathcal{G}(\mathbf{P}_i(s))~|~ i \in {1,\dots, j, \dots, N}\}$.
 \item $\max\{gap(\mathcal{G}(\mathbf{\acute{P}}_j(s)), \mathcal{G}(\mathbf{\acute{P}}_i(s))~|~ i \in {1,\dots, j, \dots, N}\}$ is closer to zero.
\end{enumerate}
where $\mathbf{\acute{P}}_j(s)=\mathbf{W_{ot}}(s)\mathbf{P}_j(s)\mathbf{W_{in}}(s)$ and $\mathbf{\acute{P}}_i(s)=\mathbf{W_{ot}}(s)\mathbf{P}_i(s)\mathbf{W_{in}}(s)$.
\begin{figure}[h!]
	\centering
	\includegraphics[width=3.1in, height=2.3in]{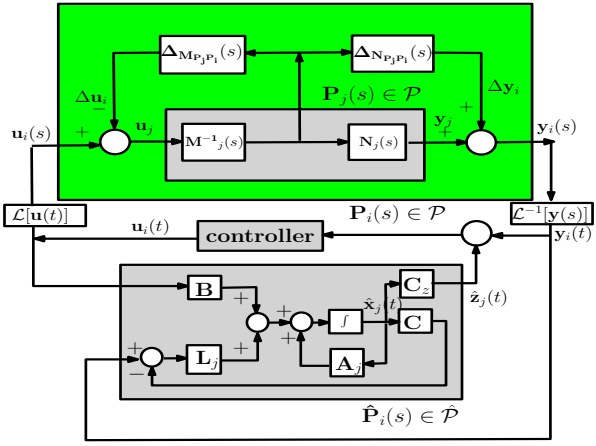}
	\caption{Simultaneous estimation process with $\mathbf{\hat{P}}_j(s)$ }
	\label{sest1}
\end{figure}

In the next section, we present the synthesis of GRMERS estimator for a typical unstable and highly coupled plants of a NAV.

\section{Synthesis of Gap Reduced Minimum Error Robust Simultaneous Estimator}\label{SRS}
\noindent The  Gap Reduced Minimum Error Robust Simultaneous (GRMERS) estimator  comprises of a MERS estimator and the GR compensators. Here,  the MERS estimator performs robust simultaneous estimation with minimal largest worst-case error as it satisfies \eqref{cond3}.  
The  GR compensators  reduces the estimation errors of MERS estimator by minimizing the \textit{gap} between the respective \textit{graphs} of the $N$ linear plants of the NAV. This minimization makes the inputs, $\mathbf{u}_i(t)$ and $\mathbf{y}_i(t)$ (ref. Fig. (\ref{sest1})), to the estimator  from the $N$ plants similar to $\mathbf{u}_j(t)$ and $\mathbf{y}_j(t)$, respectively. Hence, the GR compensators reduces the estimation errors arising from the differences in the system matrices of plants and estimator  as indicated by (\ref{ed4}). In this section,  we first explain the procedure for synthesising  the MERS estimator  model. Before  carrying out this procedure, as preliminaries, the effect of $\mathbf{v}(t)$ and $\mathbf{x}_i(t)$  on the estimation error is analyzed first.

\subsection{Preliminaries: Analysis of the Effect of Measurement  Noise and  State Vector on Estimation Error Dynamics}

\noindent Here, the effects of $\mathbf{v}(t)$ and $\mathbf{x}_i(t)$ on the estimation error  are briefly analyzed. For that,  consider   $\mathbf{\hat{P}}_l(s) \in \hat{\mathcal{P}}$ as the simultaneous estimator for estimating $\mathbf{z}_i(t)$ of $\mathbf{P}_i(s) \in \mathcal{P}$ for all $i \in \{1, \dots, l, \dots, N\}$. We now define $\mathbf{e}(s)_{\mathbf{P}_i\mathbf{\hat{P}}_l} \forall~i \in \{1,\dots, l, \dots, N\}$ using (\ref{spes_1}) and $\mathbf{u}_e(t)=-\mathbf{L}_l \mathbf{e}_{\mathbf{y}_{il}}(t)$ as
\begin{align}
\begin{split}
\mathbf{e}(s)_{\mathbf{P}_i \mathbf{\hat{P}}_l} & =    - C_z (s\mathbf{I}-(A_l-B_e \mathbf{L}_l C ))^{-1} \Delta A_{il}  \\ & \quad -C_z (s\mathbf{I}-(A_l-B_e \mathbf{L}_l C))^{-1}B_e \mathbf{L}_l ; \\ & \quad \qquad \forall ~  i \in \{1,\dots, l, \dots, N\}
\end{split}
\label{etf3}
\end{align}
In (\ref{etf3}), $- C_z (s\mathbf{I}-(A_l-B_e \mathbf{L}_l C ))^{-1} \Delta A_{il} $ is the transfer function matrix from $\mathbf{x}_i(s)$  to $\mathbf{e}_{z_{il}}(s)$ and $-C_z (s\mathbf{I}-(A_l-B_e \mathbf{L}_l C))^{-1}B_e \mathbf{L}_l $ is the transfer function matrix from $\mathbf{v}(s)$  to $\mathbf{e}_{z_{il}}(s)$. Now, when $\mathbf{\hat{P}}_l(s)$ estimates $\mathbf{z}_l(t)$ of $\mathbf{P}_l(s)$, then $\Delta A_{ll}$ is a null matrix and $\mathbf{x}_i(t)$=0. Following this, if $\mathbf{d}(t)=\mathbf{v}(t)=$0 and all the eigenvalues of $[A_l-B_e \mathbf{L}_l C]$ belong to $\mathbb{C}_-$, then $\lim_{t \to\infty} \mathbf{e}_{\mathbf{z}_{ll}}(t)=0$. This indicates that when any estimator, say $\hat{\mathbf{P}}_i(s) \in \hat{\mathcal{P}}$, estimates $\mathbf{z}_i(t)$ of $\mathbf{P}_i(s) \in \mathcal{P}$, then $\lim_{t \to\infty} \mathbf{e}_{\mathbf{z}_{ii}}(t)=0 $  if $\mathbf{v}(t)$=0 and $\hat{\mathbf{P}}_i(s) \in \mathcal{RH_\infty}$. However, when $\mathbf{\hat{P}}_l(s)$ estimates $\mathbf{z}_l(t)$ of $\mathbf{P}_l(s)$ with $\mathbf{v}\neq0$ and $\mathbf{x}_i(t) = 0$, then $\mathbf{e}_{\mathbf{z}_{ll}}(s) = -C_z (s\mathbf{I}-(A_l-B_e \mathbf{L}_l C))^{-1}B_e \mathbf{L}_l \mathbf{v}(s)$. Subsequently, in time domain, $\mathbf{e}_{\mathbf{z}_{ll}}(t) = \int_0^t -C_z e^{((A_l-B_e \mathbf{L}_l C)(t-\tau))}B_e \mathbf{L}_l \mathbf{v}(\tau)d\tau$. This integral will never be  zero when $\mathbf{v}(t)\neq$0 indicating that when $\mathbf{v}(t)\neq$0,  the estimation errors do not converge to zero during the simultaneous estimation process.  The  same phenomenon will be there even  when $\mathbf{x}_i(t) \neq$ 0.
Equation (\ref{etf3}) indicates that  $-C_z (s\mathbf{I}-(A_l-B_e \mathbf{L}_l C))^{-1}B_e \mathbf{L}_l$ are  same for all $ i \in \{1,\dots, l, \dots, N\}$. This indicates that   the effect of $\mathbf{v}(t)$ on the estimation errors remains identical when  an estimator  belonging to $\hat{\mathcal{P}}$ performs simultaneous estimation. However, (\ref{etf3}) shows that $- C_z (s\mathbf{I}-(A_l-B_e \mathbf{L}_l C ))^{-1} \Delta A_{il}~ \forall ~i \in \{1, \dots, l, \dots, N\}$ become different when $\Delta A_{il} ~\forall ~i \in \{1, \dots, l, \dots, N\}$ are distinct. This proposes that   the effect of $\mathbf{x}_i(t)$ on the estimation errors may be dissimilar when  an estimator  belongs to $\hat{\mathcal{P}}$ executes simultaneous estimation. Applying $H_\infty$ norm on both sides of (\ref{etf3}) and then using triangle  and Cauchy-Schwarz inequalities, the  $|| \mathbf{e}(s)_{\mathbf{P}_i \mathbf{\hat{P}}_l}||_\infty~\forall~i \in \{1, \dots, l, \dots, N\}$ can be written as
	\begin{align}
	\begin{split}
	|| \mathbf{e}(s)_{\mathbf{P}_i \mathbf{\hat{P}}_l}||_\infty & \leq    ||- C_z (s\mathbf{I}-(A_l-B_e \mathbf{L}_l C ))^{-1}||_\infty  ||\Delta A_{il}||_\infty  \\ & \quad+ ||-C_z (s\mathbf{I}-(A_l-B_e \mathbf{L}_l C))^{-1}B_e \mathbf{L}_l  ||_\infty;   \\ & \qquad \text{$\forall~ i \in \{1,\dots, l, \dots, N\}$}
	\end{split}
	\label{cd2}
	\end{align}
Equation (\ref{cd2}) indicates that the estimation error due to $\mathbf{x}_i(t)$   increases with the increase of  $||\Delta A_{il}||_\infty$. Next,   the development of a sufficient condition for the existence of a RS estimator based on above arguments is presented.

\subsection{Minimum Error Robust Simultaneous Estimator}
In this  subsection, we describe the development of the MERS estimator model and the MERSE algorithm. At first, the sufficient  condition for the existence of a robust simultaneous  estimator is derived.

\subsubsection{Sufficient Condition for the Existence of a Robust Simultaneous Estimator}
We now state the sufficient condition for the existence of a RS estimator  through the following theorem.
\newtheorem{theorem}{Theorem}[section]
\begin{theorem}
Given  $\mathcal{P}$,  $\hat{\mathcal{P}}$, and  $\gamma$. Consider $\mathbf{\hat{P}}_l(s) \in \hat{\mathcal{P}}$   as the simultaneous estimator   of $\mathbf{P}_i(s) \in \mathcal{P}~ \forall~i \in \{1, \dots, k, \dots,  l, \dots, N\}$. Let  $\mathbf{P}_{k}(s) \in \mathcal{P}$    satisfies the condition given by
\begin{align}
|| \Delta A_{kl} ||_\infty=\max \{|| \Delta A_{il} ||_\infty\}~\forall~i \in \{1, \dots, k, \dots, l, \dots, N\}
\label{cd4} 
\end{align}
Then, $\mathbf{P}^l_{k}(s)=\mathbf{P}_{k}(s)$ and the sufficient condition for the existence of  $\mathbf{\hat{P}}_l(s)$ as the  RS estimator   of all the plants of $\mathcal{P}$ is given 
by
\begin{align}
|| \mathbf{e}(s)_{\mathbf{P}^l_{k} \mathbf{\hat{P}}_l}||_\infty < \gamma
\label{SC1} 
\end{align}
\label{thm}
\end{theorem}

\begin{proof}
Using  (\ref{cd2}), let  	$|| \mathbf{e}(s)_{\mathbf{P}_{i} \mathbf{\hat{P}}_l}||_\infty~\forall~i \in \{1, \dots, l, \dots, N\} \setminus \{k\}$ are expressed as

\begin{small}
		\begin{equation}
	\begin{split}
	|| \mathbf{e}(s)_{\mathbf{P}_i \mathbf{\hat{P}}_l}||_\infty & \leq    ||- C_z (s\mathbf{I}-(A_l-B_e \mathbf{L}_l C ))^{-1}||_\infty  ||\Delta A_{il}||_\infty  \\ & \quad+ ||-C_z (s\mathbf{I}-(A_l-B_e \mathbf{L}_l C))^{-1}B_e \mathbf{L}_l  ||_\infty  \\ & \quad ;  \text{$\forall~ i \in \{1,\dots, l, \dots, N\} \setminus \{k\}$}
	\end{split}
	\label{cd5}
	\end{equation}
\end{small}

\noindent Likewise, 	$|| \mathbf{e}(s)_{\mathbf{P}_{k} \mathbf{\hat{P}}_l}||_\infty$ is given by

\begin{small}
	\begin{equation}
	\begin{split}
	|| \mathbf{e}(s)_{\mathbf{P}_k \mathbf{\hat{P}}_l}||_\infty & \leq    ||- C_z (s\mathbf{I}-(A_l-B_e \mathbf{L}_l C ))^{-1}||_\infty  ||\Delta A_{kl}||_\infty  \\ & \quad+ ||-C_z (s\mathbf{I}-(A_l-B_e \mathbf{L}_l C))^{-1}B_e \mathbf{L}_l  ||_\infty 
	\end{split}
	\label{cd6}
	\end{equation}
\end{small}
	
\noindent Subtracting (\ref{cd6}) from (\ref{cd5})	results in 
	
\begin{small}
	\begin{equation}
	\begin{split}
	|| \mathbf{e}(s)_{\mathbf{P}_i\mathbf{\hat{P}}_l}||_\infty-	|| \mathbf{e}(s)_{\mathbf{P}_k \mathbf{\hat{P}}_l}||_\infty & \leq    \Xi \big[ ||\Delta A_{il}||_\infty-||\Delta A_{kl}||_\infty \big] ;   \\ & \quad   \text{$\forall~ i \in \{1,\dots, l, \dots, N\} \setminus \{k\}$}
	\end{split}
	\label{cd7}
	\end{equation}
\end{small}		

\noindent where $ \Xi=||- C_z (s\mathbf{I}-(A_l-B_e \mathbf{L}_l C ))^{-1}||_\infty$. In (\ref{cd7}), RHS is negative because $\Xi > 0$ and  $\big[ ||\Delta A_{il}||_\infty-||\Delta A_{kl}||_\infty \big] ~\forall~ i \in \{1,\dots, l, \dots, N\} \setminus \{k\}$ are negative. The later terms are negative as $\mathbf{P}_k(s)$ satisfies (\ref{cd4}). Since, RHS of  (\ref{cd7}) is negative, we can rewrite (\ref{cd7}) as
	\begin{equation}
	\begin{split}
	|| \mathbf{e}(s)_{\mathbf{P}_i\mathbf{\hat{P}}_l}||_\infty-	|| \mathbf{e}(s)_{\mathbf{P}_k \mathbf{\hat{P}}_l}||_\infty & \leq -\Lambda_i  \\ & \quad   \text{$\forall~ i \in \{1,\dots, l, \dots, N\} \setminus \{k\}$}
	\end{split}
	\label{cd8}
	\end{equation}
where $\Lambda_i$ is a positive constant. Note that $|| \mathbf{e}(s)_{\mathbf{P}_i\mathbf{\hat{P}}_l}||_\infty > 0~ \forall~i \in \{1,\dots, l, \dots, N\} \setminus \{k\}$ and 	$|| \mathbf{e}(s)_{\mathbf{P}_k \mathbf{\hat{P}}_l}||_\infty > 0$. Consequently,  the condition given in (\ref{cd8}) ensures  the condition given by 
	\begin{equation}
	\begin{split}
	|| \mathbf{e}(s)_{\mathbf{P}_i\mathbf{\hat{P}}_l}||_\infty~\forall~ i \in \{1,\dots, l, \dots, N\} \setminus \{k\} < || \mathbf{e}(s)_{\mathbf{P}_k \mathbf{\hat{P}}_l}||_\infty 
	\end{split}
	\label{cd9}
	\end{equation}
\noindent Equation \eqref{cd9}	implies  $\mathbf{P}^l_k(s)=\mathbf{P}_k(s)$ and $\parallel \mathbf{e}_{\mathbf{P}^l_k\mathbf{\hat{P}}_l}(s)\parallel_\infty=\parallel \mathbf{e}_{\mathbf{P}_k\mathbf{\hat{P}}_l}(s)\parallel_\infty$ when $\mathbf{P}_k(s)$ satisfies (\ref{cd4}). Then, from (\ref{cond1}), the sufficient condition for the existence of  $\mathbf{\hat{P}}_l(s)$  as the  simultaneous estimator of  $N$ number of  plants belonging to $\mathcal{P}$ is  (\ref{SC1}).
Now consider the case of RS estimation. For that, let    $\mathbf{\bar{P}}_i(s) \notin \mathcal{P}~ \forall~ i \in \{1,\dots,k, \dots, l, \dots, N\}$ are the perturbed   plants of  $\mathbf{P}_i(s) \in \mathcal{P}~\forall~ i \in \{1,\dots,k, \dots, l, \dots, N\}$, respectively. These  perturbed plants arises due to the perturbations  in the system matrix  of $\mathbf{P}_i(s) \in \mathcal{P}~\forall~ i \in \{1,\dots,k, \dots, l, \dots, N\}$. Following this, the system matrix of $\mathbf{\bar{P}}_i(s)$ be $\bar{A}_i=A_l+\Delta A_i$ with $|| \Delta A_i||_\infty < || \Delta A_{kl}||_\infty$. Here,  $\Delta A_i$ is the bounded perturbation of   $A_i$. Because of $|| \Delta A_i||_\infty \leq || \Delta A_{kl}||_\infty$,  it obvious that $|| \mathbf{e}(s)_{\mathbf{\bar{P}}_i \mathbf{\hat{P}}_l}||_\infty < || \mathbf{e}(s)_{\mathbf{P}^l_k \mathbf{\hat{P}}_l}||_\infty~\forall~i \in \{1,\dots, k, \dots, l, \dots, N\}$. So, if  $|| \mathbf{e}(s)_{\mathbf{P}^l_k \mathbf{\hat{P}}_l}||_\infty < \gamma$, then $|| \mathbf{e}(s)_{\mathbf{\bar{P}}_i \mathbf{\hat{P}}_l}||_\infty < \gamma ~\forall~i \in \{1,\dots, k, \dots, l, \dots, N\}$. Hence, concerning $\mathbf{\hat{P}}_l(s)$,  $|| \mathbf{e}(s)_{\mathbf{P}^l_{k} \mathbf{\hat{P}}_l}||_\infty < \gamma$ is the sufficient condition for the existence of $\mathbf{\hat{P}}_l(s)$ as the  RS estimator  of all the plants of $\mathcal{P}$. This establishes the proof of the theorem.
\end{proof}
Following Theorem \ref{thm}, the sufficient condition for the existence of the estimators, $\hat{\mathbf{P}}_l(s) \in \hat{\mathcal{P}}~\forall~l \in \{1, \dots, i, \dots, N\}$, as a RS estimator  is given by
\begin{align}
||\mathbf{e}(s)_{\mathbf{P}^l_{k} \mathbf{\hat{P}}_l}||_\infty < \gamma;~ \forall~l \in\{1, \dots, i, \dots, N\}
\label{cd10}
\end{align}
Furthermore, Theorem \ref{thm} suggests that the sufficient condition for  the existence of  $\mathbf{\hat{P}}_j(s)$ as the MERS estimator  is  given in  (\ref{cond3}).  Hence, to solve MERS estimation problem, we have to solve (\ref{cond3}). 

\subsubsection{Method to Solve the Minimum Error Robust Simultaneous Estimation Problem}
 We now present the method that solves the MERS estimation problem. The MERS estimator synthesis is about finding  $\mathbf{\hat{P}}_j(s)$ with a static gain $\mathbf{L} j$,  such that the condition given in (\ref{cond3}) is satisfied. Equation (\ref{cond3}) indicates that the solution of it follows by solving (\ref{cd10}). For this, the inequalities given in (\ref{cd10}) is formulated in-terms of linear matrix inequalities (LMIs) using bounded real lemma \cite{brl}. If there exists $Q_l >0 \in\mathcal{S}^{\hat{n}}$ and  $Y_l \in \mathbb{R}^{\hat{n} \times \hat{q}}$,  then from the bounded real lemma,  the LMIs corresponding to  (\ref{cd10}) for  $ 0 < \gamma  < 1$ are given by
 \begin{equation}
\begin{aligned}
\begin{bmatrix}
Q_lA_l + A_l^TQ_l-Y_lC-C^TY_l^T & Q_l \grave{B}_l-Y_l \grave{D} & C_z^T\\
*&-\gamma\mathbf{I}&0\\
*&*&-\gamma\mathbf{I}
\end{bmatrix}<0\\
Q_l>0;\\
\forall~l \in \{1,\dots, j, \dots, N\}
\end{aligned}
\label{lmi2}
\end{equation} 
where $\grave{B}_l$=[$\Delta A_{kl}$ $0$] and $\grave{D}$ =[$0$ $\mathbf{I} \in \mathbb{R}^{\hat{r} \times \hat{r}}$]. For a given $\gamma$, solve (\ref{lmi2}) for all  $Q_l$ and $Y_l$. Thereafter, recover all the estimator gain using $\mathbf{L}_l=Q_l^{-1}Y_l~\forall~l \in \{1, \dots, j, \dots, N\}$. The feasible solution of (\ref{lmi2}) establishes that there exists $N$ number of RS estimators. Now, using $\mathbf{L}_l=Q_l^{-1}Y_l~\forall~l \in \{1, \dots, j, \dots, N\}$, compute $|| \mathbf{e}(s)_{\mathbf{P}^l_{k} \mathbf{\hat{P}}_l}||_\infty~  \forall~ l \in \{1,\dots, l, \dots, N\}$  and  find its smallest value. Then, identify the RS estimator to which this smallest value belongs and  from (\ref{cond3}), this   estimator is the MERS estimator, $\mathbf{\hat{P}}_j(s)$.
If the solution of  (\ref{lmi2})  fails to satisfy the condition specified in (\ref{cd10}), then a dynamic compensator is required.\par 
 
 We now describe a method that induce the characteristics of a dynamic compensator in the error dynamics by solving  LMIs similar to  (\ref{lmi2}). In this method, the dynamic compensator is in the form of pre and post compensators that will be cascaded with the plant models belonging to $\mathcal{P}$. Now, let us consider the pre  compensator, $\mathbf{\widetilde{W}_{ei}}(s)$, and  the post compensator, $\mathbf{\widetilde{W}_{eo}}(s)$ which would be cascaded with the plants belonging to $\mathcal{P}$.  The  $\mathbf{\widetilde{W}_{ei}}(s)$ and $\mathbf{\widetilde{W}_{eo}}(s)$ are defined by (13) and (14), respectively, of \cite{jinjgcd} with $\hat{m}=3$ and $\hat{r}=5$ ((13) and (14) are also given in  supporting material). Also, the elements of $\mathbf{\widetilde{W}_{ei}}(s)$  and $\mathbf{\widetilde{W}}_{eo}(s)$   are given by (15) and (16), respectively of \cite{jinjgcd} ((15) and (16) are also given in  supporting material). The cascading of these compensators modifies the  characteristics of all the plants of $\mathcal{P}$. This in turn alters the characteristics of the estimation error dynamics. Following this and the MERS estimation problem, the objective  is to synthesize appropriate pre/post compensators along with suitable estimator gain that reduces the largest worst-case gain of the MERS estimator  ($||\mathbf{e}(s)_{\mathbf{P}^j_{k} \mathbf{\hat{P}}_j}||_\infty$) below a given  $\gamma$. The plant  realized after cascading, $\widetilde{\mathbf{P}}_i(s)$, is defined as $\widetilde{\mathbf{P}}_i(s)=\mathbf{\widetilde{W}_{eo}}(s)\mathbf{P}_i(s)\mathbf{\widetilde{W}_{ei}}(s)$. As $\mathbf{\widetilde{W}_{eo}}(s)$ is the post compensator, the $\mathbf{v}(t)$ act at its input. Then,  the output of $\mathbf{\widetilde{W}_{eo}}(s)$ from $\mathbf{v}(t)$, $\mathbf{\widetilde{v}}(t)$, is defined as $\mathbf{\widetilde{v}}(t)=\mathcal{L}^{-1}[\mathbf{\widetilde{W}_{eo}}(s)\mathbf{v}(s)](t)$. Here, $\mathcal{L}^{-1}$ denotes the inverse Laplace transform. Considering this, the state-space form  of $\mathbf{\widetilde{P}}_i(s)$ is given by
\begin{equation}
\mathbf{\widetilde{P}}_i(s):
\begin{cases}
\dot{\mathbf{\widetilde{x}}}_i(t)=&\widetilde{A}_{i}\mathbf{\widetilde{x}}_i(t)+\widetilde{B}\mathbf{\widetilde{u}}_{i}(t)\\
\mathbf{\widetilde{z}}_i(t)=&\widetilde{C}_z\mathbf{\widetilde{x}}_i(t)\\
\mathbf{\widetilde{y}}_i(t)=&\widetilde{C}\mathbf{\widetilde{x}}_i(t)+\mathbf{\widetilde{v}}(t)
\label{eqp1}
\end{cases}
\end{equation}
where $\mathbf{\bar{x}}_i(t)$=$[\begin{smallmatrix}
\mathbf{x}_{\mathbf{\widetilde{W}_{eo}}}~|&\mathbf{x}_i~|&\mathbf{x}_{\mathbf{\widetilde{W}_{ei}}}
\end{smallmatrix}]^T \in \mathbb{R}^{22}$ is the state vector. Here, $\mathbf{x}_{\mathbf{\widetilde{W}_{eo}}}(t) \in \mathbb{R}^{5}$ and $\mathbf{x}_{\mathbf{\widetilde{W}_{ei}}}(t) \in \mathbb{R}^{3}$ are the state vectors of post and pre compensators, respectively. Besides, $\widetilde{A}_{i} \in \mathbb{R}^{22 \times 22}$, $\widetilde{B} \in \mathbb{R}^{22 \times 3}$, and $\widetilde{C} \in \mathbb{R}^{5 \times 22}$ are the system matrix, the control input matrix, and the measurement matrix, respectively. Likewise, $\widetilde{C}_z$=$[\begin{smallmatrix}
\mathbf{0}_{3, 5}~|&C_z~|&\mathbf{0}_{3 , 3}
\end{smallmatrix}]$ is a constant matrix. Additionally,  $\mathbf{\widetilde{z}}_i(t)$ is the vector to be estimated. The characteristic of $\mathbf{\widetilde{x}}_i(t)$ and $\widetilde{C}_z$ suggests $\mathbf{\widetilde{z}}_i(t)$=$\mathbf{z}_i(t)$.
Now, consider a finite set, $\widetilde{\mathcal{P}} =\{\widetilde{\mathbf{P}}_i(s) \in \mathcal{RL_\infty}~|~\widetilde{\mathbf{P}}_i(s)=\mathbf{\widetilde{W}_{eo}}(s)\mathbf{P}_i(s)\mathbf{\widetilde{W}_{ei}}(s), \mathbf{P}_i(s) \in \mathcal{P}, ~ i\in\{1,\dots, l, \dots, N\} \}$. The state-space form of any plant belongs to $\widetilde{\mathcal{P}}$ is given by (\ref{eqp1}).
Now, consider $N$ number of   estimators, $\mathbf{\check{P}}_1(s), \dots, \mathbf{\check{P}}_i(s),\dots,\mathbf{\check{P}}_l(s),\dots,\mathbf{\check{P}}_N(s)$, each formed using state-space matrices of  $\mathbf{\widetilde{P}}_1(s) \in \mathcal{\widetilde{P}},\dots,\mathbf{\widetilde{P}}_i(s) \in \mathcal{\widetilde{P}}, \dots,\mathbf{\widetilde{P}}_l(s) \in \mathcal{\widetilde{P}}, \dots, \mathbf{\widetilde{P}}_N(s) \in \mathcal{\widetilde{P}}$, respectively. Let  the finite set,  $\check{\mathcal{P}}$, contains all these $N$ estimators. The state-space forms of these estimators are given as
\begin{equation}
\mathbf{\check{P}}_1(s):
\begin{cases}
\dot{\mathbf{\check{\mathbf{x}}}}_{l}(t)={}&\widetilde{A}_l\mathbf{\check{x}}_l(t)+\widetilde{B}\mathbf{\widetilde{u}}(t)+\mathbf{\check{L}}_l (\mathbf{\widetilde{y}}(t)-\widetilde{C}\mathbf{\check{x}}_l(t))\\
\mathbf{\check{z}}_{l}(t)={}&\widetilde{C}_z\mathbf{\check{x}}_{l}(t); \forall~ l \in \{1,\dots, i, \dots, j, \dots, N\}
\end{cases}
\label{eqest-2}
\end{equation}
where  $\mathbf{\widetilde{u}}(t)$ $\in   \{\mathbf{\widetilde{u}}_1(t),\dots,\mathbf{\widetilde{u}}_i(t),\dots,\mathbf{\widetilde{u}}_N(t)\}$  and $\mathbf{\widetilde{y}}(t) \in \{\mathbf{\widetilde{y}}_1(t),\dots,\mathbf{\widetilde{y}}_i(t),\dots,\mathbf{\widetilde{y}}_N(t)\}$. Here, $\mathbf{\widetilde{u}}_i(t)$ and $\mathbf{\widetilde{y}}_i(t)$ are the control input and measurement vectors of $\mathbf{\widetilde{P}}_i(s) \in \mathcal{\widetilde{P}}$, respectively. Also, $\mathbf{\check{\mathbf{x}}}_l(t) \in \mathbb{R}^{22}$ and   $\mathbf{\check{\mathbf{u}}}_l(t) \in \mathbb{R}^{3}$ is the state  vectors of $\mathbf{\check{P}}_i(s)$. Additionally, $\mathbf{\check{\mathbf{z}}}_l(t) \in \mathbb{R}^{3}$ is the estimate of $\mathbf{\widetilde{\mathbf{z}}}_l(t)$. Using Theorem \ref{thm}, the sufficient condition for the existence of estimators, $\mathbf{\check{P}}_l(s) \in \check{\mathcal{P}}~ \forall ~l \in \{1, \dots, i, \dots, j, \dots, N\}$, as a RS estimator  is given by
\begin{align}
||\mathbf{e}(s)_{\mathbf{\widetilde{P}}^l_{k} \mathbf{\check{P}}_l}||_\infty < \gamma;~ \forall~l \in\{1, \dots, i, \dots, j, \dots, N\}
\label{cd-10}
\end{align}
We solve (\ref{cd-10}) by solving the equivalent LMIs to obtain the estimator gains, $\mathbf{\check{L}}_l~\forall~ l \in \{1,\dots, i, \dots, j, \dots, N\}$, mentioned in (\ref{lmi4}). If there exists $\check{Q}_l >0 \in\mathcal{S}^{n}$ and  $\check{Y}_l \in \mathbb{R}^{n \times \hat{q}}$,  then from the bounded real lemma \cite{brl},  the LMIs corresponding to  (\ref{cd-10}) for a $ 0 < \gamma  \leq 1$ are given by
\begin{equation}
\begin{aligned}
\begin{bmatrix}
\check{Q}_l \widetilde{A}_l + \widetilde{A}_l^T\check{Q}_l-\check{Y}_l\widetilde{C}-\widetilde{C}^T\check{Y}_l^T & \check{Q}_l \check{B}_l-\check{Y}_l \check{D} & \widetilde{C}_z^T\\
*&-\gamma\mathbf{I}&0\\
*&*&-\gamma\mathbf{I}
\end{bmatrix}<0\\
\check{Q}_l>0\\
\forall~ l \in \{1,\dots, i, \dots, j, \dots, N\}
\end{aligned}
\label{lmi4}
\end{equation} 
where $\check{B}_l$=[$\Delta \widetilde{A}_{kl}$ $0$] and $\check{D}$ =[$0$ $\mathbf{I} \in \mathbb{R}^{5 \times 5}$]. Here, $\Delta \widetilde{A}_{kl}$ is the difference between the system matrices of $\widetilde{\mathbf{P}}^l_k(s)$ and $\check{\mathbf{P}}_l(s)$. Now,  for a given $\gamma$, solve (\ref{lmi4}) for all  $\check{Q}_l$ and $\check{Y}_l$. Thereafter, recover all the estimator gain using $\mathbf{\check{L}}_l=\check{Q}_l^{-1}\check{Y}_l~\forall~l \in \{1, \dots, i, \dots,  j, \dots, N\}$.

 Theorem~\ref{thm} suggests that the estimators belong to $\mathcal{\check{P}}$ become RS estimators with $\mathbf{\check{L}}_l~\forall~ l \in \{1,\dots, i, \dots, j, \dots, N\}$  attained by solving (\ref{lmi4}). Among these RS estimators, let $\mathbf{\check{P}}_j(s)$ be the MERS estimator. Then,  following (\ref{cond2}),  $\mathbf{\check{P}}_j(s)$ needs to satisfy the condition given by
\begin{align}
\begin{split}
||\mathbf{e}(s)_{\mathbf{\widetilde{P}}^j_{k} \mathbf{\check{P}}_j}||_\infty=& \min \{||\mathbf{e}(s)_{\mathbf{\widetilde{P}}^l_{k} \mathbf{\check{P}}_l}||_\infty~\forall~l\in\{1,  \dots, j, \dots, N\}  \}  \\ & \qquad  < \gamma
\end{split}
\label{cond4}
\end{align}
Assume there exists suitable $\mathbf{\widetilde{W}_{ei}}(s)$, $\mathbf{\widetilde{W}_{eo}}(s)$, $\mathbf{\check{L}}_j$ such that $\mathbf{\check{P}}_j(s)$ satisfies (\ref{cond4}). Then, the  state-space form of $\mathbf{\check{P}}_j(s)$ is given by
\begin{equation}
\mathbf{\check{P}}_j(s):
\begin{cases}
\dot{\mathbf{\check{\mathbf{x}}}}_{j}(t)={}&\widetilde{A}_j\mathbf{\check{x}}_j(t)+\widetilde{B}\mathbf{\widetilde{u}}(t)+\mathbf{\check{L}}_j (\mathbf{\widetilde{y}}(t)-\widetilde{C}\mathbf{\check{x}}_j(t))\\
\mathbf{\check{z}}_{j}(t)={}&\widetilde{C}_z\mathbf{\check{x}}_{j}(t)
\end{cases}
\label{eqest-3}
\end{equation}
Equation (\ref{eqest-3}) indicates that $\mathbf{\check{u}}(t)$ and $\mathbf{\check{y}}(t)$ are required for the implementation of $\mathbf{\check{P}}_j(s)$. But only the information of $\mathbf{u}(t)$ and $\mathbf{y}(t)$ are available. In that case, we need to obtain $\mathbf{\check{y}}_i(t)$ using $\mathbf{\check{y}}_i(t)=\mathcal{L}^{-1}[\mathbf{\widetilde{W}_{eo}}(s)\mathbf{y}(s)](t)$ and $\mathbf{\check{u}}_i(t)$ using $\mathbf{\check{u}}_i(t)=\mathcal{L}^{-1}[\mathbf{\widetilde{W}_{ei}}^{-1}(s)\mathbf{u}(s)](t)$. Hence, to realize $\mathbf{\check{P}}_j(s)$, there should exists $\mathbf{\widetilde{W}_{eo}}(s) \in \mathcal{RH}_\infty$, $\mathbf{\widetilde{W}_{ei}}(s) \in \mathcal{RH}_\infty$,  $\mathbf{\widetilde{W}_{ei}}^{-1}(s) \in \mathcal{RH}_\infty$, and $\mathbf{\check{L}}_j=\check{Q}_l\check{Y}_l$ that satisfy (\ref{cond4}). There is no closed-form solution that provides   $\mathbf{\widetilde{W}_{ei}}(s)$,  $\mathbf{\widetilde{W}_{eo}}(s)$, and  $\mathbf{\check{L}}_j$ that  satisfies (\ref{cond4}). Hence, to obtain these feasible compensators and estimator gain, an optimization problem is formulated and expressed as
\begin{equation}
\begin{aligned}
& \underset{\mathbf{\widetilde{W}_{ei}}, \mathbf{\widetilde{W}_{eo}}, 
	\mathbf{\check{L}}_{l}}{\text{minimize}}
& & J= \min \{||\mathbf{e}(s)_{\mathbf{\widetilde{P}}^l_{k} \mathbf{\check{P}}_l}||_\infty~|~l \in \{1, \dots, j, \dots, N\} \} \\
& \text{subject to}
& & \mathbf{\check{L}}_{l}=\check{Q}_l^{-1}\check{Y}_l ~\forall~l \in \{1,\dots, j, \dots, N\} \\
& & & \mathbf{\widetilde{W}_{ei}}(s) \in \mathcal{RH}_\infty,  \mathbf{\widetilde{W}_{eo}}(s) \in \mathcal{RH}_\infty, \\
& & & \mathbf{\widetilde{W}_{ei}}^{-1}(s) \in \mathcal{RH}_\infty
\end{aligned}
\label{pbm}
\end{equation}
\noindent where $\check{Q}_l$ and $\check{Y}_l$ are obtained by solving (\ref{lmi4}). In (\ref{pbm}), the performance index is the pointwise minimum of $N$ convex function. The pointwise minimum of $N$ convex function may not be convex. Hence, a genetic algorithm based iterative approch refered  as MERSE  algorithm is developed to solve the problem given in (\ref{pbm}). This algorithm has a population-based GA solver where GA employs the same steps of GA-SCP and GA-RSSD solver mentioned in \cite{jinsmc} (these steps are also given in  supporting material).

\noindent \textit{Search Variables:} The search variables of  GA solver  are the coefficients of  $\mathbf{\widetilde{W}_{ei}}(s)$ and $\mathbf{\widetilde{W}_{eo}}(s)$. The feasible values of these search variables are those which  satisfy the constraints,  $\mathbf{\widetilde{W}_{ei}}(s) \in \mathcal{RH}_\infty$,  $\mathbf{\widetilde{W}_{eo}}(s) \in \mathcal{RH}_\infty$, and $ \mathbf{\widetilde{W}_{ei}}^{-1}(s) \in \mathcal{RH}_\infty$.\\
 \textit{Fitness functions:} The fitness function of GA solver is the performance index, $J$. \\
\textit{Termination Conditions:} The iterative  terminates when the number of generation of GA exceeds its maximum value.\\
The pseudocode of MERSE  algorithm is given in \textit{Algorithm 1}. The RS estimation using  MERS estimator   is depicted in Fig. \ref{sest3}.  This consists of the MERS estimator (shown inside the orange box) obtained from \textit{Algorithm 1}. The blocks in Fig. \ref{sest3},   $\mathbf{\widetilde{W}_{eo}}(s)\stackrel{s}{=}\left[\begin{array}{c|c}A_{eo}&B_{eo}\\\hline C_{eo}&D_{eo}\end{array}\right]$ and $\mathbf{\widetilde{W}_{ei}}^{-1}(s)\stackrel{s}{=}\left[\begin{array}{c|c}\underline{A}_{ei}&\underline{B}_{ei}\\\hline \underline{C}_{ei}&\underline{D}_{e1}\end{array}\right]$ denote the state-space representation of   $\mathbf{\widetilde{W}_{eo}}(s)$ and  $\mathbf{\widetilde{W}_{ei}}^{-1}(s)$, respectively.
\begin{algorithm}[h!]
	\caption{Pseudocode of MERSE algorithm}
	\begin{algorithmic}[1]
		\State Initialize: Genetic algorithm
		\State Input: $\mathcal{P}$, maximum number of generations, and $\gamma$  
		\If{number of generation of GA solver $\leq$ maximum value}
		\State GA obtain feasible values of search variables
		\State Compute: $\mathbf{\widetilde{W}_{ei}}(s)$ and $\mathbf{\widetilde{W}_{eo}}(s)$ using (15) and (16) of \cite{jinjgcd}.
		\State Compute: $\mathbf{\check{L}}_l~ \forall~l~ \in \{1,\dots, j, \dots, N\}$ by solving N LMIs given in (\ref{lmi4}).
		\State Compute: $J$ for fitness evaluation
		\State Fitness value of GA= $J$
		\State go to 3
		\Else
		\If{Fitness value $<$ 1} 
		\State Find $\mathbf{\widetilde{P}}_j(s)$ and $\mathbf{\check{L}}_j$
		\State Output: feasible  $\mathbf{\widetilde{W}_{ei}}(s)$ and $\mathbf{\widetilde{W}_{eo}}(s)$, $\mathbf{\widetilde{P}}_j(s)$, $\mathbf{\check{L}}_j$
		\State Exit
		\Else
		\State Output: no feasible solution to (\ref{pbm})
		\State Exit
		\EndIf
		
		\EndIf
		\label{Alg-1}
	\end{algorithmic}
\end{algorithm}
\begin{figure}[h!]
	\centering
	\includegraphics[width=3.4in, height=2.2in]{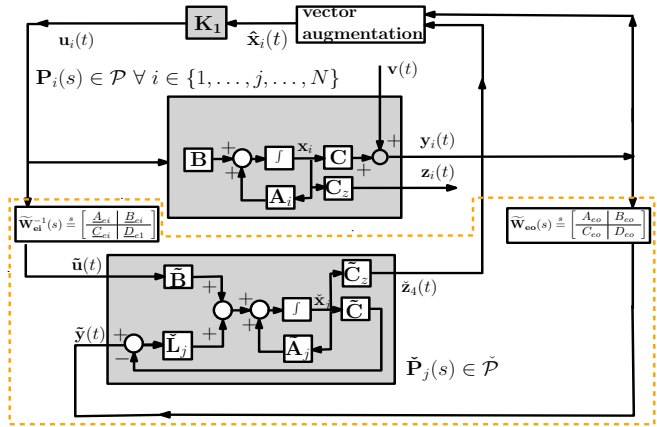}
	\caption{MERS estimator  for the NAV  (from \textit{Algorithm 1})  shown inside the orange box}
	\label{sest3}
\end{figure}
\subsection{Gap Reducing Compensators}
The development of  the GRC algorithm  is presented in this section. The gap reducing compensators obtained through the GRC algorithm are augmented with the plants in $\mathcal{P}$ in the implementation of GRMERS estimator to further reduce the estimation error obtained from the MERS estimator. The method developed in this section depends only on the plants in $\mathcal{P}$ and is independent of the MERS compensators $\mathbf{\widetilde{W}_{ei}}(s)$, $\mathbf{\widetilde{W}_{eo}}(s)$ and the gain $\mathbf{\check{L}}_j$. From the output $\mathbf{\widetilde{P}}_j(s)$ of the MERSE algorithm, the maximum value of  $\nu-$gap metric of the associated plant $\mathbf{{P}}_j(s) \in \mathcal{P}$ with the other plants in $\mathcal{P}$ is reduced further by adding suitable pre and post compensators to the plants in $\mathcal{P}$. From the definition of $\mathcal{G}(\mathbf{P}_j(s))$, making $\mathcal{G}(\mathbf{P}_i(s))~ \forall ~i \in \{1,\dots,N\}\setminus \{j\}$ closer to $\mathcal{G}(\mathbf{P}_j)$ increases closeness between $(\mathbf{u}_i , \mathbf{y}_i)~ \forall~ i \in \{1,\dots,N\} \setminus \{j\}$ and $(\mathbf{u}_j, \mathbf{y}_j)$. To make  $\mathcal{G}(\mathbf{P}_i(s))~ \forall~ i \in \{1,\dots,N\}\setminus \{j\}$ closer to $\mathcal{G}(\mathbf{P}_j(s))$, the \textit{gap} between $\mathcal{G}(\mathbf{P}_i(s))~ \forall~ i \in \{1,\dots,N\}\setminus \{j\}$ and $\mathcal{G}(\mathbf{P}_j(s))$ needs to be reduced. For developing an algorithm that minimizes the \textit{gap} between $\mathcal{G}(\mathbf{P}_i(s))~ \forall~ i \in \{1,\dots,N\}\setminus \{j\}$ and $\mathcal{G}(\mathbf{P}_j(s))$, it is necessary to compute the \textit{gap} between the \textit{graphs}. Let $\delta_v(\mathbf{P}_1(j\omega),\mathbf{P}_2(j\omega)) \in [0,1]$ be the $v$-gap metric (see \cite{jinjgcd} for details) between two plants,$\mathbf{P}_1(s)$ and $\mathbf{P}_2(s)$. Then, the \textit{gap} between two \textit{graphs} is given by \cite{vinfre}
\begin{align}
gap(\mathcal{G}(\mathbf{P}_1), \mathcal{G}(\mathbf{P}_2))=\delta_v(\mathbf{P}_1(j\omega),\mathbf{P}_2(j\omega))
\label{gap}
\end{align}
Using (\ref{gap}), \textit{gap} between two \textit{graphs} is computed. Let $\epsilon_{\mathbf{P}_j}$ denotes the maximum $v-$gap metric of $\mathbf{P}_j(s)$. Then, $\epsilon_{\mathbf{P}_j}=\max\big\{gap\big(\mathbf{P}_{j}(j\omega),\mathbf{P}_i(j\omega)\big)~\big|~ \mathbf{P}_{j}(s), \mathbf{P}_i(s)  \in \mathcal{P}~~ \forall~i \in \{1,2,\dots,N\}  \big\} $. Now, using (\ref{gap}),  $\epsilon_{\mathbf{P}_j}$ is rewritten as
\begin{align}
\begin{split}
\epsilon_{\mathbf{P}_j}=\max\big\{&\delta_v\big(\mathbf{P}_{j}(j\omega),\mathbf{P}_i(j\omega)\big)~\big|~ \mathbf{P}_{j}(s), \mathbf{P}_i(s) \\& \quad \in \mathcal{P}~~ \forall~i \in \{1,2,\dots,N\}  \big\}  
\end{split}
\label{epi}
\end{align}
Equation (\ref{epi}) suggests that  to make  $\mathcal{G}(\mathbf{P}_i)~ \forall~ i \in \{1,\dots,N\}\setminus\{j\}$ closer to $\mathcal{G}(\mathbf{P}_j)$, we need to reduce $\epsilon_{\mathbf{P}_j}$ and bring it closer towards zero.  
The maximum \textit{gap} of $\mathbf{P}_j(s)$  is  improved by cascading these models with suitable pre and post compensators, $\mathbf{W_{in}}(s) \in \mathcal{RH_\infty}$  and $\mathbf{W_{ot}}(s) \in \mathcal{RH_\infty}$, respectively \cite{jinjgcd}. Simultaneously, if required,  these  compensators can be employed to improve the frequency characteristics of the plants in $\mathcal{P}$. The basic structure of $\mathbf{W_{in}}(s)$  and  $\mathbf{W_{ot}}(s)$ are the same as that of  $\mathbf{\widetilde{W}_{in}}(s)$ and  $\mathbf{\widetilde{W}_{ot}}(s)$, respectively. However, $\mathbf{W_{ot}}(s)$ needs to be strictly proper. Let $\kappa= \big\{ \mathbf{\acute{P}}_i(s) \in \mathcal{RL_\infty}~ \big| ~  \mathbf{\acute{P}}_i(s)= \mathbf{W_{ot}}(s)\mathbf{P}_i(s)\mathbf{W_{in}}(s),~$ $\mathbf{P}_i(s)$ $\in \mathcal{P}$, $\mathbf{W_{ot}}(s)$ $\in$ $\mathcal{RH_\infty}$,   $\mathbf{W_{in}}(s)$~$\in$ $\mathcal{RH_\infty}$, $\forall~i \in \{1,2,\dots,N\}\big\}$. Now,  $\acute{\epsilon}_{\mathbf{\acute{P}}_j}$~is defined as
 \begin{align}
\begin{split}
\acute{\mathbf{\epsilon}}_{\mathbf{\acute{P}}_j}=\max\big\{&\delta_v\big(\mathbf{\acute{P}}_{j}(j\omega),\mathbf{\acute{P}}_i(j\omega)\big)~\big|~ \mathbf{\acute{P}}_{j}(s), \mathbf{\acute{P}}_i(s) \\& \quad \in \kappa, ~~ \forall~i \in \{1,2,\dots,N\}  \big\}  
 \end{split}
 \label{epi1}
 \end{align}
Then, from the GR compensators problem statement, the feasible   $\mathbf{W_{in}}(s)$ and $\mathbf{W_{ot}}(s)$ are those that achieve the following.
\begin{enumerate}
	\item  $\acute{\mathbf{\epsilon}}_{\mathbf{\acute{P}}_j} < \epsilon_{\mathbf{P}_j}$ and bring $\hat{\epsilon}_{\mathbf{\tilde{P}}_j}$ closer to zero.
	
	\item $\mathbf{W_{in}}(s)$ and $\mathbf{W_{ot}}(s)$  induce desired frequency characteristics on all the plants   belonging to $\kappa$. 
\end{enumerate}
\noindent As there does not exist any closed-form solution for the  feasible $\mathbf{W_{in}}(s)$ and $\mathbf{W_{ot}}(s)$, an optimization problem is formulated and is given by

\begin{equation}
\begin{aligned}
& \underset{\mathbf{\acute{Q}}}{\text{minimize}}
& & J_1=\acute{\epsilon}_{\mathbf{\acute{P}}_j}\\
& \text{subject to}
& & \text{1) \textit{Bound constraints on the coefficients of  pre} }\\
& & &\text{~~~~\textit{and post compensators}} \\ 
& & & \text{2) \textit{No pole-zero cancellation between} }\\
& & & \text{~~~\textit{compensators and the plants of $\mathcal{P}$}}
\end{aligned}
\label{pbm1}
\end{equation}
In    (\ref{pbm1}), $\mathbf{\acute{Q}}$  represents the set that contains the coefficients of $\mathbf{W_{in}}(s)$ and $\mathbf{W_{ot}}(s)$.
The bound constraints on the coefficients of  compensators  provide desired frequency characteristics to the plants of $\mathcal{P}$. These constraints prevent the minimization of  $J_1$  with any $\mathbf{W_{in}}(s)$ and $\mathbf{W_{ot}}(s)$ that degrade the frequency characteristics of all the augmented plants.  Note that the pre and post compensators are physically present in the closed-loop and therefore, these compensators need to be appended to the hardware. The performance index of (\ref{pbm1}) is non-convex and non-smooth. Hence, the optimization problem given in (\ref{pbm1}) is solved  using an iterative algorithm referred to as the GRC algorithm  that has a population-based genetic algorithm (GA) solver.  In that solver, GA employs the same  steps as in MERSE algorithm.\par

\noindent \textit{Search Variables:} The search variables of the GA solver  are the coefficients of $\mathbf{W_{ot}}(s)$ and $\mathbf{W_{in}}(s)$. The feasible values of these search variables are those which  satisfy all the constraints of the  problem given in (\ref{pbm1}).\\
\textit{Fitness functions:} The fitness function of the GA solver is the performance index, $J_1$, of  the optimization  problem given in (\ref{pbm1}). \\
\textit{Termination Conditions:} The iteration  terminates when the number of generations of GA solver exceeds the set maximum value.\par
\noindent The pseudocode for the iterative  algorithm is given in \textit{Algorithm 2}. 
The RS estimation using  GRMERS estimator   is depicted in Fig. \ref{sest2}.  This figure consists of the MERS estimator (shown inside the orange box) obtained from \textit{Algorithm 1}  and the gap reducing compensators (shown inside the blue box) obtained from \textit{Algorithm 2}. In Fig. \ref{sest2}, the blocks,  $\mathbf{W_{in}}(s)\stackrel{s}{=}\left[\begin{array}{c|c}A_{in}&B_{in}\\\hline C_{in}&D_{in}\end{array}\right]$ and $\mathbf{W_{ot}}(s)\stackrel{s}{=}\left[\begin{array}{c|c}A_{ot}&B_{ot}\\\hline C_{ot}&D_{ot}\end{array}\right]$ denote the state-space representation of  $\mathbf{W_{in}}(s)$ and  $\mathbf{W_{ot}}(s)$ respectively.
The  plants of  $\mathcal{P}$ are unstable and hence they are stabilized using  a full state feedback controller, $\mathbf{K}$ that  uses all the states of  $\mathbf{W_{ot}}(s)\mathbf{P}_i(s)\mathbf{W_{in}}(s)$ for feedback . These states  are obtained by augmenting   states of $\mathbf{W_{in}}(s)$ and $\mathbf{W_{ot}}(s)$ and the estimated states of $\mathbf{P}_i(s)$, $\mathbf{\hat{x}}_i(t)$ as shown in in Fig. \ref{sest2}. Here,  $\mathbf{\hat{x}}_i(t)$ is obtained by augmenting $\mathbf{{y}}_i(t)$   with $\mathbf{\widetilde{z}}_i(t)$.  Note that if we employ   MERS estimator only, then the full state vector for feedback  is $\mathbf{\hat{x}}_i(t)$.

\begin{algorithm}[h!]
	\caption{Pseudocode of GRC algorithm}
	\begin{algorithmic}[1]
		\State Initialize: Genetic algorithm
		\State Input: $\mathcal{P}$, $\mathbf{{P}}_j(s)$ (from \textit{Algorithm 1}), maximum value of generations, and $\epsilon_{\mathbf{P}_j}$.
		\If{number of generation of GA solver $\leq$ maximum value}
		\State GA obtain feasible values of search variables
		\State Compute: $\mathbf{W_{in}}(s)$ and $\mathbf{W_{ot}}(s)$
		\State Compute: $J_1$ for fitness evaluation
		\State Fitness value of GA= $J_1$
		\State go to 3
		\Else
			\If{Fitness value $<$ $\epsilon_{\mathbf{P}_j}$} 
			 \State Output: feasible compensators $\mathbf{W_{in}}$ and $\mathbf{W_{ot}}$, $J_1^*$=Fitness value 
			 \State Exit
			\Else
			  \State Output: no feasible compensators
			  \State Exit
			\EndIf		 
		\EndIf
		\label{Alg-2}
	\end{algorithmic}
\end{algorithm}
\begin{figure}[h!]
	\centering
	\includegraphics[width=3.6in, height=2.4in]{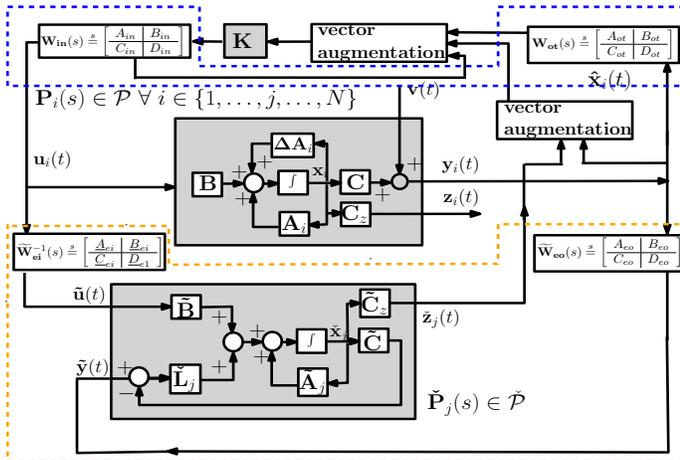}
	\caption{GRMERS estimator  for the NAV with the MERS estimator (from \textit{Algorithm 1}) and the GR compensators (from \textit{Algorithm 2}) shown inside the orange and blue boxes, respectively}
	\label{sest2}
\end{figure}

 \section{Synthesis of the GRMERS Estimator for the candidate NAV}\label{DPE}
In this section, the synthesis of a GRMERS estimator for the candidate NAV described in  Section II-III of \cite{jin}  is presented. To this end,  a GRMERS estimator is designed for four unstable MIMO  plants ($\mathcal{P}=\{\mathbf{P}_1(s),\mathbf{P}_2(s),\mathbf{P}_3(s),\mathbf{P}_4(s)\}$)    of this NAV  by designing a suitable MERS estimator and the GR compensators. The plant, $\mathbf{P}_1(s)$ is associated with the  steady turn and  climb  flight condition  at $Va$ (air speed)  of $9$~m/s, climb rate ($\dot{h}$) of $1$~m/s, turn radius ($R$) of $30$~m. Similarly, $\mathbf{P}_2(s)$, $\mathbf{P}_3(s)$ and $\mathbf{P}_3(4)$ are associated with the  flight condition  at ($Va=10$~m/s, $\dot{h}=1$~m/s, $R=30$~m), ($Va=10$~m/s, $\dot{h}=0.5$~m/s, $R=30$~m) and ($Va=10$~m/s, $\dot{h}=0$~m/s, $R=30$~m) respectively.  The state-space matrices of all the four plants including actuator dynamics are given in the supporting material. 
 \subsection{Synthesis of MERS estimator}
For synthesizing  the MERS estimator, the  MERSE algorithm  was run  with the $\gamma=1$ and the  maximum number of generations  set to 200. This algorithm  provides $\mathbf{\widetilde{P}}_j(s)=\mathbf{\widetilde{P}}_4(s)$ with $||\mathbf{e}_{\mathbf{\widetilde{P}}_{k_4} \mathbf{{\check{P}}}_{4}}(s)||_\infty=0.65$. 
The corresponding compensators $\mathbf{\widetilde{W}_{ei}}(s)$ and  $\mathbf{\widetilde{W}_{eo}}(s)$  are given by

\begin{small}
\begin{align}
 \mathbf{\widetilde{W}_{ei}}(s)=&{}&diag\bigg[\frac{47.3 s + 13.85}{s + 8.196},\frac{1.243 s + 0.6176}{ s + 2.576},\frac{0.05965 s + 1.118}{ s + 0.7056}\bigg]\label{wine1}
 \end{align}
\end{small}

 \begin{small}
 	\begin{align}
 	\begin{split}\label{wone1}
 	\mathbf{\widetilde{W}_{eo}}(s)={}&diag\bigg[\frac{323.6 s + 517.2}{s + 0.932},\frac{2447 s + 3710}{s + 0.4487},\frac{228.1 s + 163}{s + 1.478},\\& \qquad \frac{1453 s + 803.5}{s + 0.6791}, \frac{867.4 s + 564.4}{s + 0.8048}\bigg]
	\end{split}
 	\end{align}
\end{small}\noindent Now, the MERS  estimator is  $\mathbf{{\check{P}}}_4(s)$ with estimator gain , $\mathbf{\check{L}}_4 \in \mathbb{R}^{19 \times 5}$. This gain  is given in the supporting material and the state-space form of  $\mathbf{{\check{P}}}_4(s)$ is given by
\begin{equation}
\mathbf{{\check{P}}}_4(s):
 	\begin{cases}
 \dot{\mathbf{\check{x}}}_{4}={}&\widetilde{A}_4\mathbf{\check{x}}_4+\widetilde{B}_4\mathbf{\widetilde{u}}_i+\mathbf{\check{L}}_4 (\mathbf{\widetilde{y}}_i-\widetilde{C}\mathbf{\check{x}}_i )\\
 \mathbf{\check{z}}_{4}={}&\widetilde{C}_z\mathbf{\check{x}}_{4}
 \label{eqest12}
 \end{cases}
 \end{equation}
 The  dynamics of $\mathbf{\check{{P}}}_4(s)$ is asymptotically stable as  eigenvalues of $\mathbf{\check{{P}}}_4(s)$ belongs to $\mathcal{C}_-$ as shown in Fig. 3 of the supporting material.
\subsection{Synthesis of GR compensators}
The maximum $v-$gap metric of  the plant associated with the MERS estimator, $\mathbf{P}_4(s)$, is $\epsilon_{{P}_4}=0.2915$.  The GRC algorithm  was run to find the pre and post compensators that reduce the maximum $v-$gap metric associated with the $\mathbf{P}_4(s)$. These compensators  are  given by

 \begin{small}
 	\begin{align}
 	\begin{split}
 	\mathbf{W_{in}}(s)={}&diag\bigg[\frac{6723 s + 6409}{s + 1},\frac{0.00026 s + 0.06368}{ s + 2.333},\\ & \qquad \quad \frac{0.00278 s + 0.9258}{s + 0.2322}\bigg]\label{wine}
 	\end{split}
	\end{align}
 \end{small}
 \begin{small}
 	\begin{align}
 	\begin{split}\label{wone}
 	\mathbf{W_{ot}}(s)={}&diag\bigg[\frac{0.0002834}{ s + 0.3444},\frac{0.0002312}{s + 0.7408},\frac{0.003528}{ s + 57.93},\\& \qquad \frac{0.00005}{ s + 0.7941}, \frac{15.27}{ s + 2.917}\bigg]
 	\end{split}
	\end{align}
 \end{small}
Also, the obtained value of $J_1^*$ is $0.02793$ and is considerably  lower than $\epsilon_{\mathbf{P}_4}(s)=0.2915$. This suggests that the closeness between $\mathbf{\acute{P}}_4(s)$ and other plants in $\kappa$  is increased. 

\subsection{Stability,  performance, and robustness of  GRMERS  estimator}
In this subsection, a study has been conducted to evaluate the stability, performance, and robustness of the GRMERS estimator  synthesised for the plants belonging to $\mathcal{P}$. For this, the GRMERS estimator implementation   follows the architecture   shown  in Fig. \ref{sest2} with $j=4$.  To study the effectiveness of the GR compensators, we  utilize the MERS estimator  to estimate the desired states of  all the nominal plants of $\mathcal{P}$. The MERS estimator implementation  for this purpose follows the architecture   shown  in Fig. \ref{sest3} with $j=4$. It is necessary to compare the nominal and robust performance of the GRMERS estimator while estimating all the plants belonging to $\mathcal{P}$ with a benchmark estimator. If the GRMERS estimator is designed to estimate a single plant alone, then the GR compensators are not required to improve the performance. In that case, the GRMERS estimator can resemble a $H_\infty$ filter.  Following this, a single $H_\infty$ filter is designed for each nominal plant belonging to $\mathcal{P}$ (refer to the supporting material for more details). The nominal and robust performances of these individual filters while estimating the corresponding nominal and perturbed plants with measurement noise are then evaluated and compared with the performance of the GRMERS estimator while estimating all the plants belonging to $\mathcal{P}$.
The measurement noise considered in the simulations is that of the rate-gyro which is used for measuring $p$, $q$, and $r$. For simulating the effect of this noise in MATLAB, the  noise with a  zero mean, RMS (root mean squared) value of $0.06$~${}^{\circ}/s$, and a spectral density of $0.005$~$\frac{{}^{\circ}/s}{\sqrt{Hz}}$ \cite{jinthesis} is added to  $p$, $q$, and $r$ (rate-gyro output). In all  the simulations, the plants are excited with a doublet thrust input. Since all the four plants are unstable,  a static full state feedback controller was implemented first to make them stable. The controller, $\mathbf{K_1} \in \mathbb{R}^{3 \times 11}$ associated with both the MERS estimator and $H_\infty$ filter are the same. However,  the controller, $\mathbf{K} \in \mathbb{R}^{3 \times 19}$ used along with the GRMERS estimator has different dimension because of the presence of the GR compensators. The details of the controllers are given in the supporting material.
\subsubsection{Stability and nominal performance analysis}
The nominal performance of the estimators is obtained through the estimation of the nominal plants,   $ \mathbf{P}_i(s) \in \mathcal{P}~\forall~i~\in~\{1,\dots,4\}$. The estimated vectors of these plants are: $\mathbf{z}_i(t) = [u(t),\, v(t),\, w(t),\, \delta_e(t),\, \delta_T(t), \, \delta_r(t)]^T~\forall~i~\in~\{1,\dots,4\}$.
The  boundedness of the estimated vector and the estimation errors indicate that the GRMERS  and MERS estimator are stable (refer to the supporting material for the plots of the estimated vector and the estimation errors). 
In this paper, the Normalized-Root-Mean-Squared Error (NRMSE), $x^{NE}$ is used as a quantitative measure for the estimator's performance in estimating any scalar variable, $x(t)$. The NRMSE, $x^{NE}$, of $x(t)$ is defined as
\begin{equation}
  x^{NE} =\frac{\sqrt{\frac{1}{T}\sum_{\acute{t}=1}^{T}(x^{\acute{t}}-\hat{x}^{\acute{t}})}}{max(x)-min(x)}
\end{equation}
 where $T$ is the number of observation, $x^{\acute{t}}$ is the $\acute{t}$-th   observation of $x$, $\hat{x}^{\acute{t}}$ is the $\acute{t}$-th estimation of $x$, $max(x)$ is the maximum value of $x$, and $min(x)$ is the minimum value of $x$. 
Now, the normalized error vector, $\mathbf{z}_i^{e}$, of  $\mathbf{z}_i(t)$   is defined as
 \begin{equation}
   \mathbf{z}_i^{e} = [u_i^{NE},v_i^{NE},w_i^{NE},\delta_{e_i}^{NE},\delta_{T_i}^{NE},\delta_{r_i}^{NE} ]^T 
 \end{equation}
 where $u_i^{NE}$, $v_i^{NE}$, $w_i^{NE}$, $\delta_{e_i}^{NE}$, $\delta_{T_i}^{NE}$, and $\delta_{r_i}^{NE}$ are the NRMSE of $u(t)$, $v(t)$, $w(t)$, $\delta_e(t)$, $\delta_T(t)$, and $\delta_r(t)$ of $i$th plant, respectively.
 The nominal and robust performances of the robust simultaneous
estimator are acceptable if $|| \mathbf{z}_i^e||_2~\forall~i~\in~\{1,\dots,4\}$ are closer to zero. The $||\mathbf{z}_i^e||_2$ of the  estimators are given in Table \ref{table:2}. The values shown in this table suggest that the  performance of the individual $H_\infty$ filters is the best followed by the GRMERS estimator. Moreover, the values of $||\mathbf{z}_i^e||_2$ associated with the GRMERS estimator is smaller than the values of the MERS estimator as indicated by Table \ref{table:2}. Following this, the reduction in the estimation error caused by the GR compensators with reference to the MERS estimator expressed as the percentage when the GRMERS estimator estimates $\mathbf{P}_1(s)$, $\mathbf{P}_2(s)$, and $\mathbf{P}_3(s)$ are 41.13~$\%$, 55.8344~$\%$, and 47.5410~$\%$, respectively. This suggests that a GRMERS estimator, formed by integrating  GR compensators and a MERS estimator, outperforms a sole MERS estimator. Note that the estimation error reduction by reducing the \textit{gap}  between the plants using  GR compensators for  $\mathbf{P}_4(s)$  is not required  as the design of GRMERS and MERS estimators   are  based on same plant, $\mathbf{P}_4(s)$.
\begin{table}[h!]
\caption{Nominal performance comparison between GRMERS, MERS and individual $H_\infty$ filters}
\begin{center}
\begin{tabular}{ |l|l|l|l|l| }
\hline
Estimator & 
 $\mathbf{P}_1(s)$ & $\mathbf{P}_2(s)$ &$\mathbf{P}_3(s)$&$\mathbf{P}_4(s)$ \\
\cline{2-5}
& $||\mathbf{z}_1^e||_2$ & $||\mathbf{z}_2^e||_2$ &$||\mathbf{z}_3^e||_2$&$||\mathbf{z}_4^e||_2$\\
\hline
GRMERS&3.85e-2 &3.52e-2&1.92e-2&4.9e-3\\
MERS&6.54e-2 &7.97e-2&3.66e-2&6.6e-3\\
$H_\infty$ filter &3.80e-4 &1.40e-4&1.14e-4&3.82e-4\\
\hline
\end{tabular}
\end{center}
\label{table:2}
\end{table}
\subsubsection{Robust performance analysis} 
Here, the robust performance of the GRMERS estimator is presented. For this purpose,  $\mathbf{P}_1(s)$, $\mathbf{P}_2(s)$, $\mathbf{P}_3(s)$, and $\mathbf{P}_4(s)$ are perturbed into $\mathbf{\bar{P}}_1(s)$, $\mathbf{\bar{P}}_2(s)$, $\mathbf{\bar{P}}_3(s)$, and $\mathbf{\bar{P}}_4(s)$, respectively, by inducing $8.5~\%$, $13~\%$, $10~\%$, and $5~\%$ parametric uncertainties into the system matrices of $\mathbf{P}_1(s)$, $\mathbf{P}_2(s)$, $\mathbf{P}_3(s)$, and $\mathbf{P}_4(s)$ such that $||\Delta A_i||_\infty <||\Delta A_{k4}||_\infty \forall~ i \in \{1,\dots,4\}$. The  state-space matrices of the perturbed plants are given the supporting material. The  robust performance of the GRMERS estimator   is compared with the  individual  $H_\infty$ filters designed around each plant belonging to $\mathcal{P}$. A  simulation setup similar to the one explained earlier is used to study the robustness of the GRMERS and the individual $H_\infty$ filters for each plant. Table \ref{table:3} shows the robust estimation performances of all the estimators and filters considered in this paper. This table indicates  that the robust estimation performance of  GRMERS estimator is better than the individual $H_\infty$ filters as the  $||\mathbf{z_i^e}||_2$ of the GRMERS estimator is lower than the $H_\infty$ filters. The  GRMERS estimator's estimation errors  are $17.86~\%$, $40.35~\%$, $41.13~\%$, and $43.00~\%$  smaller than $H_\infty$ filters  while estimating $\mathbf{\bar{P}}_1(s)$, $\mathbf{\bar{P}}_2(s)$ , $\mathbf{\bar{P}}_3(s)$, and $\mathbf{\bar{P}}_4(s)$, respectively.
 \begin{table}[h!]
\caption{Robust performance comparison between GRMERS,    and individual $H_\infty$ filters}
\begin{center}
\begin{tabular}{ |l|l|l|l|l| }
\hline
Estimator &$\mathbf{\bar{P}}_1(s)$ & $\mathbf{\bar{P}}_2(s)$ &$\mathbf{\bar{P}}_3(s)$&$\mathbf{\bar{P}}_4(s)$ \\
\cline{2-5}
& $||\mathbf{z}_1^e||_2$ & $||\mathbf{z}_2^e||_2$ &$||\mathbf{z}_3^e||_2$&$||\mathbf{z}_4^e||_2$\\
\hline
GRMERS&5.52e-2 &5.38e-2&4.39e-2&2.73e-2\\
$H_\infty$ filter &6.76e-2 &9.02e-2&8.10e-2&4.80e-2\\
\hline
\end{tabular}
\end{center}
\label{table:3}
\end{table}

\section{Conclusion}\label{cons}
In this paper, a novel robust simultaneous estimator referred to as the GRMERS estimator has been developed to estimate the states of a finite set of unstable MIMO plants of a NAV.  This GRMERS estimator comprises of a MERS estimator and GR compensators where the former provides robust simultaneous estimation with minimal largest worst-case estimation error and the latter reduces this estimation error further by decreasing the \textit{gap} between the \textit{graphs} of $N$ linear plants. For a given set of stable/unstable plants, a sufficient condition for the existence of a MERS estimator has been presented using LMIs and robust estimation theory. Two separate non-convex tractable optimization problems, one for the solution of the sufficient conditions and the other to obtain the GR compensators, are formulated in terms of LMIs using robust estimation theory and $\nu$-gap metric, respectively. The solutions for these optimization problems are obtained using two GA-based iterative algorithms. The tractability of these algorithms is successfully demonstrated by the generation of a feasible MERS estimator and GR compensators for four unstable plants of a typical fixed-wing NAV. The simulation results highlight that the GRMERS estimator is easily implementable in a typical NAV, and its performance is within the acceptable limit. Further, the 2-norm of normalized error of the GRMERS estimator is lower than that of the MERS estimator, which indicates that the GR compensators are effective in reducing the simultaneous estimation errors. The nominal and robust performance of the GRMERS estimator is compared with the individually designed $H_{\infty}$ filters.  The performance of GRMERS and individual $H_{\infty}$ filters are evaluated for both nominal and perturbed plants, and the results indicate that the GRMERS estimator is robust under perturbation than $H_{\infty}$ filters and GRMERS estimator provides satisfactory performance for all nominal plants. The novel GRMERS estimator  is ideal for implementation in  computational-resource constrained systems.

\end{document}